\documentclass[11pt]{article}
\pdfoutput=1
\usepackage{amsfonts,amsmath, graphicx,latexsym,url,epsf,mathrsfs, bbm, paralist, comment, subfigure,amssymb, natbib, verbatim}
\usepackage[small]{caption}
\defcitealias{mturk}{\scshape MTrk}
\defcitealias{clickworker}{\scshape ClkWrkr}
\defcitealias{crowdflower}{\scshape CrdFlwr}
\defcitealias{wikihoeffding}{\scshape HffdngBnd}

\DeclareMathOperator*{\argmax}{arg\,max}

\usepackage{amsfonts,amsmath,graphicx,latexsym,url,epsf, algpseudocode, verbatim, setspace, float, paralist, mathtools, color}

\usepackage{tikz}
\usetikzlibrary{arrows,calc}
\tikzset{
>=stealth',
help lines/.style={dashed, thick},
axis/.style={<->},
important line/.style={thick},
connection/.style={thick, dotted},
}

\usepackage[linesnumbered, boxed, ruled, commentsnumbered]{algorithm2e}
\renewcommand{\algorithmcfname}{Mechanism}

\usepackage{changepage}
\usepackage[T1]{fontenc}
\usepackage[latin9]{inputenc}
\usepackage{array}
\usepackage{multirow}
\usepackage{fullpage}

\newenvironment{subproof}{\paragraph{Proof}}{\ \ \ $\blacksquare$}

\newif\ifabstract
\abstracttrue
\newif\iffull
\ifabstract \fullfalse \else \fulltrue \fi

\usepackage
  [breaklinks,bookmarks,bookmarksnumbered,bookmarksopen,bookmarksopenlevel=2]
  {hyperref}

\let\realbfseries=\bfseries
\def\bfseries{\realbfseries\boldmath}

\makeatletter
\def\GrabProofArgument[#1]{ #1: \egroup\ignorespaces}
\def\proof{\noindent\textbf\bgroup Proof%
           \@ifnextchar[{\GrabProofArgument}{. \egroup\ignorespaces}}

\makeatother

\begin{document}
\newtheorem{conj}{Conjecture}
\newtheorem{lemma}{Lemma}
\newtheorem{definition}{Definition}
\newtheorem{corol}{Corollary}
\newtheorem{theorem}{Theorem}
\newtheorem{rem}{Remark}
\newtheorem{corollary}{Corollary}
\newtheorem{proposition}{Proposition}
\newtheorem{claim}{Claim}
\newtheorem{fact}{Fact}

\def\figs{.}


\renewcommand{\algorithmcfname}{Mechanism}
\def\rbar{\ensuremath{\overline{r}}}
\def\rstar{\ensuremath{{r}^{\star}}}
\def\sstar{\ensuremath{{S}^{\star}}}
\def\over{\overline}
\def\fit{budget-tight}
\def\uopt{\ensuremath{U^{\star}}}
\def\ualg{\ensuremath{{U_\mathsf{ALG}}}}
\def\aux{\textsf{aux}}
\def\cbar{\ensuremath{\overline{c}}}
\def\cmax{\ensuremath{{c_{\textsf{max}}}}}
\def\umax{\ensuremath{{u_{\textsf{max}}}}}
\def\umin{\ensuremath{{u_{\textsf{min}}}}}
\def\bmax{\max_{i\in S}\{c_i/u_i\}}
\def\over{\overline}
\def\scale{\ref{alg.uniformratio}($f$)}
\def\r{\mathsf{r}}
\def\chibar{\over{\chi}}
\def\chidbar{\over{\over{\chi}}}
\def\F{\mathcal{F}}
\def\Scale{\textsc{Envy-Free}}
\def\Real{\mathbb{R}}
\def\bbR{\mathbb{R}}
\def\tix{\tilde{x}}
\def\tip{\tilde{p}}
\def\bbE{\mathbb{E}}
\def\calP{\mathcal{P}}
\def\np{unit-payment}
\def\Np{Unit-payment}
\def\umax{u_{\textsf{max}}}
\def\calD{\mathcal{D}}
\def\cutoff{\ensuremath{{\textsf{cutoff}}}}

\title{Mechanism Design for Crowdsourcing: An Optimal 1-1/e Competitive Budget-Feasible Mechanism for Large Markets}
\author{Nima Anari\footnote{University of California, Berkeley. \href{mailto:anari@eecs.berkeley.edu}{anari@eecs.berkeley.edu}},
Gagan Goel\footnote{Google Research, New York. \href{mailto:gagangoel@google.com}{gagangoel@google.com}},
Afshin Nikzad\footnote{Stanford University, Stanford.
\href{mailto:nikzad@stanford.edu}{nikzad@stanford.edu}
}}

\date{}


\maketitle

\begin{abstract}

In this paper we consider a mechanism design problem in the context of large-scale crowdsourcing markets such as Amazon's Mechanical Turk \citepalias{mturk}, ClickWorker~\citepalias{clickworker}, CrowdFlower~\citepalias{crowdflower}. In these markets, there is a requester who wants to hire workers to accomplish some tasks. Each worker is assumed to give some utility to the requester on getting hired. Moreover each worker has a minimum cost that he wants to get paid for getting hired. This minimum cost is assumed to be private information of the workers. The question then is - if the requester has a limited budget, how to design a direct revelation mechanism that picks the right set of workers to hire in order to maximize the requester's utility?

We note that although the previous work (\cite{2010-focs_singer_budget-feasible-mechanisms, 2011-soda_improved_budget_feasible}) has studied this problem, a crucial difference in which we deviate from earlier work is the notion of {\em large-scale} markets that we introduce in our model. The notion of a large-scale market that we consider is a natural one which states that the (private) cost of each worker is {\em small} compared to the budget of the requester. Without the large market assumption, it is known that no mechanism can achieve a competitive ratio better than $0.414$ and $0.5$ for deterministic and randomized mechanisms respectively (while the best known deterministic and randomized mechanisms achieve an approximation ratio of $0.292$ and $0.33$ respectively). In this paper, we design a budget-feasible mechanism for large markets that achieves a competitive ratio of $1-1/e \simeq 0.63$. Our mechanism can be seen as a generalization of an alternate way to look at the {\em proportional share} mechanism, which is used in all the previous works so far on this problem. Interestingly, we can also show that our mechanism is optimal by showing that no truthful mechanism can achieve a factor better than $1-1/e$; thus, fully resolving this setting. Finally we consider the more general case of submodular utility functions and give new and improved mechanisms for the case when the market is large.

\end{abstract}




\maketitle

\pagebreak

\tableofcontents
\newpage
\onehalfspacing

\setcounter{page}{1}


\section{Introduction}
Crowdsourcing is a recent phenomenon that is used to describe the procurement of a large number of workers to do certain tasks. These tasks can be of a variety of natures and - to give a few examples - include image annotation, data labeling for machine learning systems, consumer surveys, rating search engine results, spam detection, product reviews, etc. There are several platforms (such as Amazon's Mechanical Turk \citepalias{mturk}) that facilitate and automate various steps involved in setting up and executing crowdsourcing tasks.

A key challenge in these online labor markets is to be able to properly price the tasks. Since the requester (the one who wants to procure workers) is usually budget constrained, pricing the tasks too high can result in lower output for the requester. On the other hand, pricing the tasks too low can disincentivize workers to work on the tasks. This makes pricing a non-trivial step for the requester when setting up a crowdsourcing task. One idea - to make pricing more automated and to prevent economic loss from poor pricing - is to design a direct revelation mechanism that solicits bids from workers to report their cost of participation, and based on this decide which workers to hire and how much to pay them. 


A simple model that captures the above problem is as follows: There is a set $S$ of workers. Worker $i$ has a private cost $c_i$ and provides utility $u_i$ to the requester on getting hired. We want to design a truthful mechanism that decides which workers to recruit and how much to pay them. The goal is to maximize the requester's utility without violating her budget constraint.

For the above model, ~\cite{2010-focs_singer_budget-feasible-mechanisms} gave an incentive-compatible mechanism that achieves an approximation ratio of $1/6$ compared to the {\em offline optimum} that knows the costs of the workers. Later on ~\cite{2011-soda_improved_budget_feasible} improved the approximation ratio to $\frac{1}{2 + \sqrt(2)} \simeq 0.292$ (and to $1/3$ for randomized mechanisms). Chen et. al. also showed that no deterministic mechanism can achieve an approximation ratio better than $\frac{1}{1+\sqrt(2)} \simeq 0.414$, and no randomized mechanism can achieve an approximation ratio better than $0.5$.

Our work is motivated by the following observation: Most of the crowdsourcing tasks are {\em large-scale} in nature in terms of the number of workers involved. On the other hand if one looks at the impossibility result of ~\cite{2011-soda_improved_budget_feasible}, they involve only a small number of workers (specifically, only  3 workers). Thus, this leads to a natural open question - {\em Do these lower bounds extend to the case of large markets? or can one design better mechanisms for this important case of large markets?}

In this paper, we seek to understand the above question. We show that one can significantly improve the approximation ratio for the case of large markets. We give a mechanism that achieves an approximation ratio of $1-1/e \simeq 0.63$ for large markets. In addition, we show that our mechanism is the best possible mechanism by showing that no truthful budget-feasible mechanism can achieve a factor better than $1-1/e$. Finally, we look at the more general case of submodular utility functions.






\subsection{The Model}
We define the model abstractly: Consider a reverse auction scenario with one buyer and $n$ sellers, where the set of sellers is denoted by $S$. 
Each seller $i\in S$ owns a single item (denoted by item $i$) and has a {\em private} cost $c_i$ for it. The buyer derives a utility of $u_i$ from item $i$. The buyer has a limited budget $B$, and its goal is to buy a subset of items that maximizes her utility without exceeding her budget.

Note that if the sellers are not strategic and the costs are known to the buyer, then this is the well-known {\em knapsack} optimization problem.
However, the cost $c_i$ is assumed to be a private information of seller $i$. Thus we are interested in designing direct-revelation mechanisms where the buyer solicits bids from the sellers, and then computes which sellers to buy from and how much to pay them. More formally, a mechanism $\mathcal{M}$ consists of two functions $A: (\bbR_+)^n \rightarrow \{0,1\}^n$ and $P: (\bbR_+)^n \rightarrow (\bbR_+)^n$. The allocation function $A(\cdot)$ takes as input the costs of $n$ sellers and reports the set of winners. The payment function $P(\cdot)$ takes as input the costs of $n$ sellers and reports how much each seller should pay. Sometimes we will use functions $A_i: (\bbR_+)^n \rightarrow \{0,1\}$ (and similarly $P_i(\cdot)$) for each $i\in S$ to refer to the restriction of functions $A(\cdot)$ and $P(\cdot)$ for seller $i$. \\

The mechanism $\mathcal{M} = (A, P)$  should satisfy the following properties:
\begin{enumerate}
\item Budget Feasibility: The sum of the payments made to the sellers should not exceed $B$, i.e., $\sum_i P_i(\mathbf{c})  \leq B$ for all $\mathbf{c}$.
\item Individual rationality: A winner $i\in S$ is paid at least $c_i$.
\item Truthfulness/Incentive-Compatibility: Reporting the true cost should be a dominant strategy of the sellers, i.e. for all non-truthful reports $\over{c}_i$ from seller $i$, it holds that
\begin{align*}
 P_i(\over{c}_i, \mathbf{c_{-i}}) - c_i\cdot A_i(\over{c}_i, \mathbf{c_{-i}})   \leq  P_i(c_i, \mathbf{c_{-i}}) - c_i \cdot A_i(c_i, \mathbf{c_{-i}})
\end{align*}
\end{enumerate}

Among all mechanisms that satisfy the above properties, we are interested in the ones that give high utility to the buyer. Note that no mechanism can achieve utility larger than $U^*(\mathbf{c, u})$, where $U^*(\mathbf{c, u})$ (or simply $U^*$ for brevity) is the utility of the knapsack optimization problem assuming costs of the sellers are known to the buyer. We say a mechanism $\mathcal{M}$ is an $\alpha$-approximation (for $\alpha \leq 1$) if it gives utility at least $\alpha \cdot U^*(\mathbf{c, u})$ for any $\mathbf{c}$ and $\mathbf{u}$. 

{\bf Indivisible vs Divisible Items.} Note that the above description is given for indivisible items, however, we can define the above problem for divisible items as well. For instance, if the item being sold by a seller is his own time, then it can be modeled as a divisible item. For fraction $x  \leq 1$ of a divisible item,  the cost of seller $i$ is $x.c_i$ and the utility obtained by the buyer is $x. u_i$.  The allocation function for divisible items is defined as $A: (\bbR_+)^n \rightarrow [0,1]^n$.


{\bf More general utility functions.} An interesting generalization of the above model is when the utility function over the set of items is a submodular function rather than additive functions.  We denote this function by $U: 2^S \rightarrow \bbR_+$  (for additive functions, $U(T) = \sum_{i \in T} u_i$, for $\forall T \subseteq S$). We assume that the utility function is known to the buyer.




\subsubsection{{\bf The Large Market Assumption}} \label{sec.linearlargemarket}
Crowd-sourcing systems are excellent examples of {\em large markets}. Informally speaking, a market is said to be large if the number of participants are large enough that no single person can affect the market outcome significantly. Our results take advantage of this nature of the crowdsourcing markets to give better mechanisms.

We define the {\bf large market assumption} as follows: We assume that in our model, the cost of a single item is very small compared to the buyer's budget $B$. More formally, let $\cmax = \max_{i\in S}\{c_i\}$. Then, the large market assumption is defined as below.\\

{{\bf The Large Market Assumption:}} $\cmax \ll B$.
\\

In other words, we define the {\em largeness ratio} of the market to be $\theta=\frac{\cmax}{B}$ and analyze our mechanisms for $\theta\to 0$.

This assumption - also known as the {\em small bid to budget ratio assumption} - is used in other large-market problems as well (for instance, see \cite{MehtaSVV07} for a similar definition with application in online advertising).
All the mechanisms that we present in the main body of the paper (mechanisms for additive utility functions) will be analyzed under this assumption. 
The mechanisms that we design for submodular utility functions work under a different large market assumption which is explained below.

\paragraph{An Alternative Assumption}
We also suggest another definition for large markets, the discussion of which will be deferred to the appendix. Our mechanisms for submodular utility functions work under this assumption; moreover, we can slightly modify our mechanisms for additive utility functions so that they work under this assumption as well, while preserving their approximation ratio. We define this assumption below.

Let $\umax =\max_{i\in S}{u_i}$ and $U^*$ be the total utility of the optimum solution (i.e. the maximum utility that the buyer can achieve when the costs are known to her). This large market assumption states that:
\\

{{\em The Alternative Large Market Assumption:}}  $\umax \ll U^*$.
\\

In other words, we define the largeness ratio of the market to be $\theta=\max_{i\in S} \frac{u_i}{U^\star}$ and analyze our mechanisms for when $\theta\to 0$.

We note that our impossibility result for additive utilities (Section \ref{sec.hardness}) holds for either of the two definitions.

\subsection{Our Results}
In this paper, we design optimal budget-feasible mechanisms for {\em large markets}. To the best of our knowledge, we are the first ones to study the case of large markets. We list our results below:

\begin{enumerate}
\item If the items are divisible, we design a deterministic mechanism which satisfies all the required properties and has an approximation ratio of $1-1/e$ (Section \ref{sec.optimumf}). Note that previously, no mechanism was known for the case of divisible items. In fact, one can show that no bounded approximation ratio is possible for divisible items if the large market assumption is dismissed. 

\item If the items are indivisible, we can modify our mechanism and give a randomized truthful mechanism for this case which achieves an approximation ratio of $1-1/e$. (Section \ref{sec.rounding})

\item In Section \ref{sec.hardness}, we show that the above results are optimal by proving that no truthful (and possibly) randomized mechanism can achieve approximation ratio better than $1-1/e$. Our hardness result holds for both cases of divisible and indivisible items.


\item For the case of submodular utility functions, we design deterministic mechanisms that achieve approximation ratios of $\frac{1}{2}$ and $\frac{1}{3}$ with exponential and polynomial running times respectively. Note that we only consider the case of indivisible items for submodular utility functions. (Section \ref{sec.submod})
\end{enumerate}

As we saw in Section \ref{sec.linearlargemarket}, one could define a notion of $\theta$-large market, i.e. a market with largeness ratio $\theta$. To gain a better understating of the problem, we focus on large markets (i.e. when $\theta\to 0$) and state our main theorems for this setting. However, our mechanisms do not need ``very large'' markets to perform well; for instance, in the knapsack problem with additive utilities, the approximation ratio \footnote{we didn't try to optimize the dependence on $\theta$ in our analysis as we focus on the main ideas for the sake of better understanding.} is $(1-1/e)\cdot (1-6\theta/5)$ when all the items have equal utilities (Section \ref{SEC.TRUTHU1}). 
Thus, say for $\theta = 1/20$ and $\theta=1/40$ (which are reasonable assumptions in many settings) we get approximation factors $0.592$ and $0.613$ respectively.

Also we point out that the above results have applications beyond crowdsourcing - for instance, see \cite{2011-wsdm_singer_how-to-win-friends-and-influence-people} for application in marketing over social networks, and \cite{2103_corr_budget_experimental_design} for application in experiment design. \cite{2011-wsdm_singer_how-to-win-friends-and-influence-people} provides a truthful mechanism with approximation ratio $\approx 0.032$ and \cite{2103_corr_budget_experimental_design} provides an approximately truthful mechanism with approximation ratio $\approx 0.077$.  For both these settings, large market assumption is a very reasonable assumption to make; thus, our results apply to these applications as well. In particular, our results give fully truthful mechanisms for these applications with approximation ratios $\frac{1}{2},\frac{1}{3}$ (for exponential and polynomial running time respectively) in large markets. 


\subsection{Related Work}
The most relevant related work is that of \cite{2010-focs_singer_budget-feasible-mechanisms} and \cite{2011-soda_improved_budget_feasible}. \cite{2010-focs_singer_budget-feasible-mechanisms} first introduced this model (without the large market assumption). For the case of additive utilities and indivisible items, he gave a deterministic mechanism with an approximation ratio of $1/6$.  \cite{2011-soda_improved_budget_feasible} later improved it to $1/(2+ \sqrt{2})$, and also gave a randomized mechanism with an approximation ratio of $1/3$. They gave a lower bound of $1/(1 + \sqrt{2})$ and $1/2$ for deterministic and randomized mechanisms respectively.  For the case of submodular utilities,  \cite{2010-focs_singer_budget-feasible-mechanisms} gave a randomized mechanism with an approximation ratio of $1/112$ which was improved to $1/7.91$ by \cite{2011-soda_improved_budget_feasible}. \cite{2011-soda_improved_budget_feasible} also gave an exponential time deterministic mechanism for submodular utility functions with an approximation ratio of $1/8.34$. 

\cite{DobzinskiPS11} looked at the more general sub-additive utility functions and gave a $1/log^2(n)$ and $1/log^3(n)$ approximation ratio for randomized and deterministic mechanisms respectively. \cite{singla13incentives} design budget feasible mechanisms for adaptive submodular functions with applications in community sensing.

In another work, \cite{2012-stoc_budget_feasible_prior_free_to_bayesian} study this problem in the bayesian setting. \cite{2011-wsdm_singer_how-to-win-friends-and-influence-people} looks at the application of this model in marketing over social networks. ~\cite{2103_corr_budget_experimental_design} study the application of this model in experiment design.

Another related model that has been inspired from crowdsourcing applications is when the workers arrive online. A sequence of papers model this as an online learning problem. See \cite{2013-www_truthful-incentives, 2012-ec_singer_budget-feasible-posted-prices, 2013-www_pricing_mechanisms} for more details.

Finally, we note that our assumption for large markets is similar to the assumption made in other application areas; notably in the Adwords problem as studied by \cite{MehtaSVV07}. See ~\cite{gagan-mehta, 2009-ec_adwords, FeldmanHKMS10, FeldmanKMMP09} for other models motivated by online advertising where they make similar assumptions.

\subsection{Roadmap} \label{sec.organization}
The readers are encouraged to read this section before proceeding further.
We begin by developing some intuition in Section \ref{sec.intuition}. In Section \ref{sec.uniscalef},
a simple {\em proportional share mechanism} which forms the basis for \cite{2010-focs_singer_budget-feasible-mechanisms,
2011-soda_improved_budget_feasible} is introduced. The mechanism picks a single cutoff for
the utility to cost ratio in such a way that the whole budget is consumed. In Section \ref{sec.nonuniscalef}, we generalize the simple {\em proportional share
mechanism} to a class of mechanisms parameterized by a single-variable allocation function. In later sections, we show that this generalization improves the approximation ratio \footnote{Although the truthfulness is sacrificed, later we augment the mechanism so that it becomes truthful without compromising the approximation ratio in
large market.}.
We develop some intuition by considering a simple instance of our generalized mechanisms: instead of a hard cutoff that is used in the {\em proportional share mechanism}, i.e. a two-level allocation rule, we consider
a special class of three-level allocation rules and show that they can improve the approximation
ratio.

The generalized mechanisms introduced are not in general truthful. In Section \ref{sec.optmec}, we introduce
a simple method to make them truthful, while maintaining their individually rational and budget-feasibility. Later when we introduce the optimal mechanism, we show that its approximation ratio does not get
compromised by utilizing this method in large markets.

In Section \ref{sec.optimumf}, we find one of the generalized mechanisms which provides an approximation ratio of $1-1/e$ in large markets. 
In Section \ref{sec.hardness}, we complement this result by showing that no 
truthful mechanism can achieve approximation ratio better than $1-1/e$.
In Section \ref{sec.rounding}, we adapt our mechanism to the case of indivisible items.


In Section \ref{sec.submod}, we present two mechanisms for submodular utility functions which have exponential and polynomial running times and approximation ratios $\frac{1}{2}$ and $\frac{1}{3}$, respectively. 

\section{A Simple $\frac{1}{2}$-Approximate Truthful Mechanism} \label{sec.uniscalef}
In this section, we briefly explain the previous mechanism designed for this problem for the additive utility functions that gives an approximation factor of $\frac{1}{2}$ in large markets. 

\begin{definition}
{\em Cost-per-utility rate} of a seller $i$ is equal to $c_i/u_i$. 
\end{definition}

A natural approach to this problem tries to find a single {\em payment-per-utility rate} (denoted by rate $r$) at which all the winning sellers get paid. In other words, this approach picks a single number $r$ and makes a payment of $r.u_i$ to seller $i$ if she wins and pays her $0$ otherwise. For brevity, we sometimes call the payment-per-utility rate $r$ simply the {\em rate $r$} when there is no risk of confusion.

Individual rationality implies that a seller $i$ is willing to sell her item at rate $r$ iff $r\geq c_i/u_i$.
 Initially the buyer declares a very large rate $r$, and then sees which sellers are willing to sell at this rate. If the total cost to buy from all these sellers at rate $r$ is higher than the budget $B$, then the buyer decreases the rate $r$. More formally, a natural descending price auction for this problem works as follows: 
\begin{enumerate}
\item Let $A$ denote the set of active sellers, and initially set $A=S$. \item Start with a very high rate $r$. 
\item Verify if all the active sellers can be paid with rate $r$, i.e. whether $ \sum_{i\in A} r.u_i \leq B$ or not. 
\item If the payment is feasible, then allocate the subset $A$, make the payment and stop. 
\item If the payment is not feasible then decrease $r$ slightly; update $A$ accordingly by removing the sellers $i$ for whom $c_i/u_i > r$; go to Step 3.
\end{enumerate}

The above auction captures the main idea behind the {\em proportional share} mechanisms designed in \cite{2010-focs_singer_budget-feasible-mechanisms, 2011-soda_improved_budget_feasible}\footnote{
It is worth pointing out that for submodular utilities, they need to use an additional trick: constructing a (sorted) list of sellers in a greedy manner before running the auction.}, although they describe it in a forward auction format. It is not hard to see that the above mechanism is truthful, budget-feasible, and in large markets achieves an approximation ratio of $\frac{1}{2}$ (with small modifications, this can be converted to a randomized $\frac{1}{2}$-approximation for arbitrary markets as well \cite{2011-soda_improved_budget_feasible}).

\section{Our Approach} \label{sec.intuition}

In this section, we give a high level overview of our approach. Sections \ref{sec.allocationrule} and \ref{sec.paymentrule} are preliminary sections and must be read before proceeding further. The notions defined in these sections are explained more intuitively using an example in Section \ref{sec.example}.
Also for rest of the paper we will assume that the sellers' items are divisible, unless we explicitly talk about indivisible items.

\subsection {A Notion of An Allocation Rule} \label{sec.allocationrule}
To build our new ideas, we first introduce and formalize the notion of an allocation rule.

An allocation rule $f:\mathbb{R}^+\rightarrow [0,1]$ is a function which determines how much to buy from a given seller. The domain of allocation rules is the {\em cost per utility rate}; meaning, given a $(u_i, c_i)$ pair of seller $i$, the allocation rule $f$ says that we should buy $f(\frac{c_i}{u_i})$ fraction of seller $i$'s item. We do not enforce using the same allocation rule for all sellers.

We say an allocation rule $f: \Real^+\rightarrow [0,1]$ is a {\bf \em Standard Allocation Rule} if $f$ is a decreasing function such that $f(0)=1$ and $f(e-1)=0$.\footnote{The choice of $e-1$ is just for simplifying the future calculations; it can be replaced with any other constant.}

For any standard allocation rule $f:\Real_+\to [0,1]$, we can define an associated family of allocation rules 
\begin{align*}
\F(f)= \left\{f_r: \Real_+\to [0,1]\right\}_{r>0}
\end{align*}
 where $f_r$ denotes an allocation rule which is same as $f$ except that it is stretched along the horizontal axis with ratio $r$, i.e. $f_r(x)=f(x/r)$ for all $x\geq 0$.

As we will see later, any single standard allocation rule $f$ and its corresponding family of allocation rules $\F(f)$ will uniquely specify our mechanism. At a high level, our mechanism will work as follows: we will pick the largest positive $r$ such that $f_r \in \F(f)$ is budget-feasible; meaning the sum of the payments with allocation rule $f_r$ does not exceed $B$. However, note that we have not yet defined a payment rule given an allocation rule $f_r$ - we define it next.

\subsection{Payment Rule} \label{sec.paymentrule}
Recall that given a function $f_r \in \F(f)$ and $(u_i, c_i)$ pair for a seller $i$, the value of $f_r(\frac{c_i}{u_i})$ only tells us what fraction of seller $i$'s item to buy. But how much should we pay seller $i$ in order to give incentives to seller $i$ to report its cost truthfully to the mechanism? We compute these payments based on the well known Myerson's characterization of the truthful mechanisms \cite{myersonslemma}.

Let the payment rule for seller $i$ is denoted by $P_{i,r}:\bbR_+\to\bbR_+$ for an allocation rule $f_r$. Here $P_{i,r}$ maps the reported cost of seller $i$ into its payment. 

To define $P_{i,r}$, we do the following thought process: Let's divide seller $i$'s item into $u_i$ different pieces. Note that now the seller's cost for each piece is $\frac{c_i}{u_i}$. Thus function $f_r$ can now be seen as mapping the cost of a single piece into the fraction of that piece that we will buy. Let  $Q_r (x) :\bbR_+\to\bbR_+$ denote the function that maps the cost for a single piece into a payment for that piece. Now, Myerson's characterization \cite{myersonslemma} says that the payment for each piece is given by the following formula:

\begin{align*}
Q_r (x)=x\cdot f_r(x) + \int_{x}^{\infty} f_r(y) \operatorname{d}y.
\end{align*}

\begin{figure} 
\center
 \begin{tikzpicture}[scale=1] 
    \coordinate (y) at (0,3);
    \coordinate (x) at (5,0);
    \draw[axis] (y) -- (0,0) --  (x);
    
 \draw[black, ultra thick, domain=0:4] plot (\x, {1.2*ln(5-\x)});
\draw[dashed] (2.2,0) -- (2.2,{1.2*ln(5-2.2)});
\draw[ultra thick][dashed] (0,{1.2*ln(5-2.2)-0.01}) -- (2.2,{1.2*ln(5-2.2)-0.01});
\fill[fill=gray, opacity=0.2] (0,0) -- plot [domain=0:2.2] (\x,{1.2*ln(5-2.2)}) -- (2.2,0)  -- cycle;
\fill[fill=gray, opacity=0.2]  (2.2,0) -- plot [domain=2.2:4] (\x,{1.2*ln(5-\x)}) -- (4,0)  -- cycle;

    \node at (0,3.3) {$f_r$};
    \node at (-0.2,2) {$1$};
    \node at (2.3,-0.2) {${c_i}/{u_i}$};
 \end{tikzpicture}
 \caption{With allocation rule $f_r$, the payment to seller $i$ is defined to be $u_i$ times the shaded area under the curve.}
 \label{fig.payment}
 \end{figure}
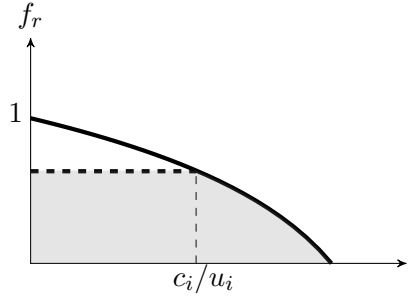

Intuitively, $Q_r (x)$ represents the area under the curve as seen in Figure~\ref{fig.payment}. Going forward, we will call the function $Q_r$ a {\em unit-payment rule}. Note that $P_{i,r}(x)$ and  $Q_r$ are related by the following formula:
\begin{align*}
P_{i,r}(x) = u_i \cdot Q_r(x/u_i),
\end{align*}

Thus, to summarize, for an allocation rule $f_r$, we buy $f_r(c_i/u_i)$ units of her item, and pay her $P_{i,r}(c_i)$ amount of money.

{\em Remark}: We make a remark that the above payment rule is truthful only if the allocation rule $f_r$ that is offered to seller $i$ doesn't depend on the private information (cost $c_i$ in this case) of the seller $i$. If the allocation rule $f_r$ does depend on the private information of the seller, then the mechanism may or may not be truthful. In the next section we give a mechanism in which the allocation rule $f_r$ for a seller $i$ depends on its reported costs. Later in section \ref{sec.truthfulmechanisms}, we give a mechanism where the allocation rule $f_r$ for a seller $i$ doesn't depend on its reported cost, thus our payment rule will ensure that the resulting mechanism is truthful.

\subsection{Example} \label{sec.example}
Suppose the buyer has budget $B=13/3$, $S=\{s_1,s_2\}$ and sellers $s_1,s_2$ each own an item with cost $2$ and $4$, respectively. Also, suppose both of the items have utility $1$. 

Let the mechanism use the family of curves $\mathcal{F}(f)$ for $f(x)=1-x$ where the domain of $f$ is $[0,1]$. The mechanism should find the largest $r$ for which $f_r$ is budget feasible. To this end, we set $r$ to be a very large number and decrease $r$ until $f_r$ is budget feasible. For instance, suppose we start from $r=10$ (see Figure \ref{fig.infeasible}). Observe that the payment of the mechanism in this case would be $4.8$ and $4.2$ to $s_1$ and $s_2$, respectively. Since the sum of payments exceed $13/3$, then $r=10$ is not budget feasible. Consequently, we reduce $r$ further until the mechanism becomes budget feasible at $r=4$: At this point, the payment of the mechanism to $s_1$ and $s_2$ would respectively be $8/3$ and $5/3$, which sum up to be exactly $B$. (see Figure \ref{fig.feasible})
\begin{figure}[h!]
  \centering
    \includegraphics[width=0.3\textwidth]{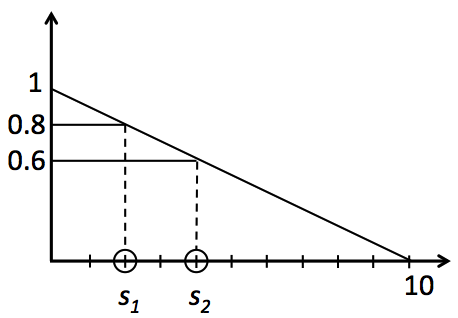}
      \caption{The mechanism is not budget feasible at $r=10$.}
      \label{fig.infeasible}
\end{figure}

\begin{figure}[h!]
  \centering
    \includegraphics[width=0.3\textwidth]{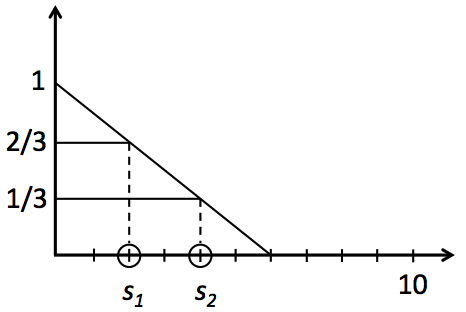}
      \caption{The mechanism is budget feasible at $r=6$.}
            \label{fig.feasible}
\end{figure}

\subsection{First Attempt: A Parameterized Class of Envy-Free Mechanisms} \label{sec.envyfree}
In this section we describe a mechanism (denoted by  Mechanism \scale{}) that is not always truthful, but it will form the basis of our truthful mechanism. Moreover, some structural results about this mechanism will be useful while analyzing our truthful mecahanism, thus we will be talking about this mechanism throughout the paper. This mechanism is parameterized by the choice of a standard allocation rule $f$. The mechanism described in this section {\em offers} a single allocation rule $f_r \in \F(f)$ to all the sellers, thus it is envy-free (although it may not be truthful).

\begin{definition}
We say that an allocation rule $f_r$ is a {\em budget-feasible} allocation rule if $\sum_{i\in S} P_{i,r} (c_i) = B$, i.e. the payments defined with respect to $f_r$ sum up to $B$.
\end{definition}

Now given any standard allocation rule $f$, the mechanism starts with a very large scaling ratio $r=\infty$ so that we are guaranteed to have $\sum_{i\in S} P_{i,r} (c_i) > B$. 

Then, the mechanism decreases $r$ until the rule $f_r$ becomes a {\em budget-feasible} rule (say at $r = \rstar$). The mechanism stops at this point and uses $f_{\rstar}$ and $\{P_{i,\rstar}\}_{i\in S}$ to determine the allocations and payments. The ratio $\rstar$ is also called the {\em stopping rate} of the mechanism. 
We define this process formally in Mechanism \scale{}.

\begin{algorithm}
\SetAlgoRefName{Envy-Free}
\LinesNotNumbered
\SetKwInOut{Input}{input}
\SetKwInOut{Output}{output}
\Input{Budget $B$, $(u_i, c_i)$ pair for each seller $i$}
\Output{A scaling ratio $\rstar$}
\BlankLine
$r\leftarrow \infty$\;
\While {$f_r$ is not a budget-feasible rule}{
	Decrease $r$ slightly\;
}
$\rstar\leftarrow r$\;
Output the scaling ratio $\rstar$\;
\caption{Parameterized by a standard allocation rule $f$}
\label{alg.uniformratio}
\end{algorithm}

One can easily see that the above mechanism is budget-feasible, individually rational, and envy-free; however, it may not be truthful. Also, the efficiency of the above mechanism depends on the choice of function $f$. Thus, an important question is: What is the optimal choice of function $f$? Let's first understand the performance of the above mechanism for a simple choice of function $f$.

\begin{definition}
A standard allocation rule $f:\mathbb{R}^+\rightarrow [0,1]$ is called a uniform standard allocation rule if $f(x) = 1$ for $x < e - 1$, and $f(x) = 0$ otherwise. Figure \ref{fig.uniform} depicts this curve.
\end{definition}

\begin{figure}
\center
 \begin{tikzpicture}[scale=1] 
    \coordinate (y) at (0,3);
     \node at (0,3.2) {\tiny \textsf{Allocation}}; 
    \coordinate (x) at (6,0);
     \node at (6.2,-0.2) {\tiny \textsf{cost per utility rate}}; 
    \draw[axis] (y) -- (0,0) --  (x);
    \draw [ultra thick] (0,2) -- (4,2)  -- (4,0);
    \node at (-0.2,2) {$1$};
    \node at (4,-0.2) {$e-1$};    

 \end{tikzpicture}
 \caption{The Uniform Standard Allocation Rule}
  \label{fig.uniform}
\end{figure}
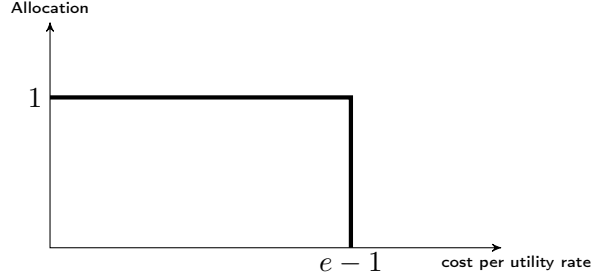

One can show that the above envy-free mechanism when run using a uniform standard allocation rule, mimics the simple factor $\frac{1}{2}$ mechanism presented earlier. Thus, it turns out to be truthful as well for this choice of standard allocation rule. However for more general allocation rules, the above envy-free mechanism might not be truthful. Thus before we answer the harder question about the optimal choice of function $f$, we next describe the truthful version of the above envy-free mechanism.

\section{A Parameterized Class of Truthful Mechanisms}\label{sec.truthfulmechanisms}

We use a simple trick to convert Mechanism \scale{} to a truthful mechanism. The idea is to define, for each seller $i$, an allocation rule which does not depend on $c_i$. In particular, we define the allocation rule for seller $i$ to be $f_{r_i}$, where $r_i$ will be chosen independently of $c_i$. 
For finding $r_i$, we run Mechanism \scale{} on the instance which is obtained by setting $c_i$ to be $0$ while keeping cost of the other sellers intact; $r_i$ would be the stopping rate of the mechanism \scale{}. The formal definition of the truthful mechanism appears in Mechanism \ref{alg.truthful}.

\begin{algorithm}

\SetAlgoRefName{Truthful$(f)$}
\LinesNotNumbered
\SetKwInOut{Input}{input}
\SetKwInOut{Output}{output}
\Input{Budget $B$, $(u_i, c_i$) pair for each seller $i$}
\BlankLine
\ForEach{$i\in S$}{
	$temp\leftarrow c_i$\;
	$c_i\leftarrow 0$\;
	$r_i\leftarrow $ Mechanism \ref{alg.uniformratio}$(f)$\;
	$c_i\leftarrow temp$\;
}
\ForEach{$i\in S$}{
Allocate $f_{r_i}(c_i)$ from seller $i$\;
Pay $P_{i,r_i}(c_i)$ to seller $i$\;
}
\caption{Parameterized by a standard allocation rule $f$ }
\label{alg.truthful}
\end{algorithm}

In Lemma \ref{lem.scalecanbegood}, we prove that Mechanism \ref{alg.truthful} is individually rational, truthful, and budget-feasible for any given standard allocation rule $f$. First, we state the following useful lemma.
\begin{lemma} \label{lem.rdec}
For any seller $i\in S$ we have $\rstar\geq r_i$.
\end{lemma}
\begin{proof}
The proof is based on the fact that $P_{i,r}(x)$ is an increasing function of $r$ (for a fixed $x$) and 
is a decreasing function of $x$ (for a fixed $r$). 
The proof is by contradiction, suppose $\rstar < r_i$. 
Let $c'_j =c_j$ for all $j\in S\backslash\{i\}$ and let $c'_i=0$. Observe that 
\begin{align*}
B = \sum_{j\in S} P_{j,\rstar}(c_j) &\leq \sum_{j\in S} P_{j,\rstar}(c'_j)\\
&<  \sum_{j\in S} P_{j,r_i}(c'_j) 
\end{align*}
where the first inequality is due to the fact that $P_{j,\rstar}(x)$ is a decreasing function of $x$ and the second inequality is due to the fact that $\rstar < r_i$. However, note that the above inequalities imply that $B < \sum_{j\in S} P_{j,r_i}(c'_j)$, which contradicts with the budget feasibility of Mechanism \scale{}: see that $\sum_{j\in S} P_{j,r_i}(c'_j)$ represents the payment of \scale{} when the costs are $c'_1,\ldots,c'_n$, and so it can not be larger than $B$.
\end{proof}

\begin{lemma} \label{lem.scalecanbegood}
Mechanism \ref{alg.truthful} is individually rational, truthful, and budget-feasible.
\end{lemma}
\begin{proof}
Note that the allocation and payment rules for seller $i$, i.e. $f_{r_i}, P_{i,r_i}$, do not depend on the cost reported by her. 
This fact, along with the fact that $f_{r_i}$ is a monotone rule (decreasing function) implies individual rationality and truthfulness. The proof is almost identical to the proof of Myerson's Lemma and we do not repeat it here. 

The proof for budget feasibility needs a bit more work. 
Let $p_i, p'_i$ denote the payments to seller $i$ respectively in Mechanism \ref{alg.truthful} and Mechanism \scale, i.e. $p_i = P_{i,r_i}(c_i)$ and $p'_i=P_{i,\rstar}(c_i)$.
The lemma is proved if we show that $p_i\leq p'_i$, since we have $\sum_{i\in S} p'_i = B$. 

To see $p_i\leq p'_i$, note that $P_{i,r}(x)$ is an increasing function of $r$ (for a fixed $x$). So, since we have $\rstar\geq r_i$ due to Lemma \ref{lem.rdec}, it must be the case that $P_{i,r_i}(c_i) \leq P_{i,\rstar}(c_i)$. 
\end{proof}

\section{A ($1-1/e$)-Approximate Optimal Truthful Mechanism} \label{sec.optimumf}
So far, we have introduced a parameterized class of individually rational, truthful, and budget-feasible mechanisms for the problem: Passing any standard allocation rule $f$ to Mechanism \ref{alg.truthful} fixes the mechanism which we denote by $\mathcal{M}_f$. 
Our goal in this section is to find the {\em most efficient} mechanism in this class. Formally, given a standard allocation rule $f$, we denote the approximation ratio of $\mathcal{M}_f$ by $\mathcal{R}_f$ and define it as:
\begin{align*}
\mathcal{R}_f=\inf_{{I}} \frac{\mathcal{U}_f ({I})}{\mathcal{U}^{\star}(I)},
\end{align*}
where the infimum is taken over all instances $I$ of the problem~\footnote{If we are focused on large markets, we take only instances $I$ for which the largeness ratio is smaller than some threshold, and take the limiting approximation factor as the threshold goes to $0$.} . Here $\mathcal{U}_f ({I})$ denotes the utility obtained by $\mathcal{M}_f$ in instance $I$, and $\mathcal{U}^{\star}(I)$ denotes the optimum utility in instance~$I$.

The most efficient allocation rule $f$, is the one which maximizes $\mathcal{R}_f$. Our goal, in this section and Section \ref{sec.hardness}, is to find the most efficient allocation rule and its corresponding approximation ratio. Formally, we prove the following theorem.

\begin{theorem} \label{thm.main}
The most efficient standard allocation rule for Mechanism \ref{alg.truthful} is $f(x)=\ln(e-x)$, for which we get $\mathcal{R}_f=1-1/e$, i.e. it has an approximation ratio $1-1/e$.
\end{theorem}

We prove this theorem in two parts: In the first part we show that $\mathcal{R}_f\geq 1-1/e$ for $f(x)=\ln(e-x)$; this is proved in the current section. In the second part, we show that $\mathcal{R}_g\leq 1-1/e$ for any standard allocation rule $g$. This fact can be seen as a consequence of our hardness result in Section \ref{sec.hardness}, 
which states that no truthful mechanism can achieve approximation ratio better than $1-1/e$. We also provide a more direct (alternative) proof in Section \ref{sec.lb} that shows our choice of $f(x)=\ln(e-x)$ is optimal among all possible choices of the standard allocation rules.


\subsection{Finding an optimal $f$ for the (non-truthful) Mechanism \scale}
In this section, we prove that Mechanism \scale{} has approximation ratio $1-1/e$ for $f(x)=\ln(e-x)$. 
Note that the Mechanism \scale{} is not truthful, however its analysis will be helpful when analyzing our truthful mechanism in Section \ref{SEC.TRUTHU1} and in Section \ref{sec.generalu}.
Here, we analyze Mechanism \scale{} assuming that the true costs are known; later, in Section \ref{SEC.TRUTHU1}, we use this result to prove that Mechanism \ref{alg.truthful} has approximation ratio $1-1/e$ for the same choice of $f$. 

\subsubsection{Preliminaries}

We use $g_r$ to denote the inverse of an allocation rule $f_r$, i.e. $g_r(x)=f_r^{-1}(x)$.
Given an allocation rule $f_r$, we also write an alternative definition of its corresponding \np{} rule $Q_r$. This definition, rather than being in terms of $\frac{c_i}{u_i}$, would be in terms of $f_r(\frac{c_i}{u_i})$. This alternative definition is denoted by $G_r$, and is defined such that $Q_r(\frac{c_i}{u_i}) = G_r(f(\frac{c_i}{u_i}))$. For instance, if a seller owns an item with utility $1$, then we pay her $G_r(x)$ when a fraction $x$ of her item is allocated.
To be more precise, for $y=f_r(\frac{c_i}{u_i})$ we define
\begin{align*}
G_r(y)=\int_0^y g_r(x)\operatorname{d}x=Q_r(\frac{c_i}{u_i}).
\end{align*}
We also denote $g_1$ and $G_1$ respectively by $g$ and $G$.

\begin{proposition}
Given the standard allocation rule $f(x)=\ln(e-x)$, it is straight-forward to verify that $g(x)=e-e^x$ and $G(x)=ex-e^x+1$. Also, $f_r(x)=\ln\left(\frac{er-x}{r}\right)$.
\end{proposition}
From now on in this section, we assume that $f(x)=\ln(e-x)$.
Next, we prove a useful inequality in the following lemma which will be used in the analysis of \scale. 

\begin{lemma} \label{lem.key}
For any $x,\alpha$ such that $0\leq x,\alpha \leq 1$ we have
\begin{align*}
G(x)-\alpha\cdot g(x) \leq e\cdot (x-\alpha\cdot (1-{1}/{e})).
\end{align*}
\end{lemma}
\begin{proof}
\begin{align*}
& \ \alpha (e^x-1) \leq e^x -1 \\
\Rightarrow &\ 
\alpha e^x -e^x +1 \leq \alpha\\
\Rightarrow &\ 
\alpha e^x -e^x +1 + e(x-\alpha)  \leq \alpha+ e(x-\alpha)\\
\textrm{\small (by the definition of $g, G$)\ }\Rightarrow &\ 
G(x)-\alpha\cdot g(x) \leq e\cdot (x-\alpha\cdot (1-{1}/{e})).
\end{align*}
\end{proof}

\subsubsection{Approximation Ratio of Mechanism \scale}
In the following lemma, we prove the efficiency of Mechanism \scale{} when all sellers report true costs.
\begin{lemma} \label{lem.key2approx}
If sellers report true costs, then Mechanism \scale{} has approximation ratio $1-{1}/{e}$.
\end{lemma}
\begin{proof}
Observe that w.l.o.g. we can assume $\rstar=1$: If $\rstar\neq 1$, then we can construct a new instance which is {\em similar} to the original instance and has stopping rate $1$. More precisely, there exists some $\beta>0$ such that if we multiply the budget and the reported costs by $\beta$, the stopping rate becomes equal to $1$. Note that this operation will not change the optimal solution or the solution of \scale{} and can be performed w.l.o.g.

Now, suppose that a fraction $x_i$ of item $i$ is allocated by \scale. 
Since $\rstar =1$, we can use Lemma \ref{lem.key} to write the following set of inequalities:
\begin{align*}
G(x_i)-\alpha_i \cdot g(x_i) \leq e\cdot (x_i-\alpha_i\cdot (1-{1}/{e})) &\hspace{1cm} \forall i\in S,
\end{align*}
where $\alpha_i$ is the fraction that is allocated from seller $i$ in the optimal solution (recall that we are are comparing \scale{} with the optimum fractional solution).
The above inequalities can be multiplied by $u_i$ on both side and be written as:
\begin{align*}
u_i\cdot \left(G(x_i)-\alpha_i \cdot g(x_i)\right) \leq u_i\cdot e\cdot (x_i-\alpha_i\cdot (1-{1}/{e})) &\hspace{1cm} \forall i\in S.
\end{align*}
By adding up these inequalities, we get:
\begin{align}
\sum_{i\in S}u_i\cdot \left(G(x_i)-\alpha_i \cdot g(x_i)\right) \leq  e\cdot \sum_{i\in S} u_i\cdot \left(x_i-\alpha_i\cdot (1-{1}/{e})\right). \label{eq.addall}
\end{align}

Now, we show that if
\begin{align}
0\leq \sum_{i\in S}u_i\cdot \left(G(x_i)-\alpha_i \cdot g(x_i)\right),\label{eq.addallrhs}
\end{align}
then the lemma is proved using \eqref{eq.addall} and \eqref{eq.addallrhs}. First we show why \eqref{eq.addall} and \eqref{eq.addallrhs} prove the lemma, and then in the end, we prove \eqref{eq.addallrhs} itself.

Observe that \eqref{eq.addall}
and \eqref{eq.addallrhs} imply that 
\begin{align}
0 \leq  \sum_{i\in S} u_i\cdot\left(x_i-\alpha_i\cdot (1-{1}/{e})\right). \label{eq.2stepsb4}
\end{align}
Now, let  $U$ denote the utility gained by \scale{} and $\uopt = \sum_{i\in S}u_i \alpha_i$ denote the utility of the optimum (fractional) solution; see that \eqref{eq.2stepsb4} implies
\begin{align*}
(1-1/e)\cdot \uopt = \sum_{i\in S}\alpha_i u_i \cdot (1-{1}/{e}) \leq  \sum_{i\in S}x_i u_i = U,
\end{align*}
This would prove the lemma. 

So, it only remains to show that \eqref{eq.addallrhs} holds: First observe that $\sum_{i\in S}u_i\cdot G(x_i)=B$, since the sum represents the payment of \scale. Also, see that $\sum_{i\in S}\alpha_i u_i \cdot g(x_i) \leq B$, since this sum is a lower bound on the cost of the optimal solution, which is at most $B$. 
\end{proof}

\subsection{Special case: Analyzing our truthful mechanism for unit utilities}\label{SEC.TRUTHU1}
In this section, we prove that Mechanism \ref{alg.truthful} has approximation ratio $1-1/e$ in large markets when it uses the standard allocation rule $f(x)=ln(e-x)$ for the special case when all the utilities are equal to $1$. In other words, we will show that approximation ratio approaches $1-{1}/{e}$ as $\theta$, the market's largeness ratio, approaches $0$ for the case of unit utilities. The proof for the case of general utilities is intricate and appears in Section \ref{sec.generalu}.

Note that the assumption of unit utilities imply $c_i/u_i=c_i$ for any seller $i$. Next, we state two lemmas before proving the approximation ratio. For simplicity in the analysis, w.l.o.g., assume that $c_1 \leq c_2\leq \ldots \leq c_n$.

\begin{lemma}  \label{lem.rissorted}
$r_1 \geq r_2 \geq \ldots \geq r_n$. 
\end{lemma}
 
\begin{lemma} \label{lem.concave}
Let $u^\star(b):\bbR_+\to\bbR_+$ denote the maximum utility that the buyer can achieve with budget $b$ (when the items are divisible). Then, $u^\star(b)$ is a concave function.
\end{lemma}

Proofs for both of these lemmas are straight-forward and are deferred to the appendix, Section \ref{sec.defferedproofs}.

\begin{lemma}
Mechanism \ref{alg.truthful} has approximation ratio $1-{1}/{e}$ when all the items have utility equal to $1$.
\end{lemma}
\begin{proof}
Recall that $\uopt=u^\star(B)$ and let $U$ denote the utility achieved by Mechanism~\ref{alg.truthful}. We need to show that $(1-1/e) \cdot \uopt \leq U$.
Instead of showing that $U=\sum_{i\in S} f_{r_i}(c_i)$ is large enough compared to $\uopt$, we show that $\sum_{i\in S} f_{r_n}(c_i)$ is large enough compared to $\uopt$; the lemma then would be proved since we have $f_{r_n}(c_i) \leq f_{r_i}(c_i)$ for all $i\in S$. To see why $f_{r_n}(c_i) \leq f_{r_i}(c_i)$, it is enough to note that $r_n\leq r_i$ due to Lemma \ref{lem.rissorted} which implies $f_{r_n}(c_i) \leq f_{r_i}(c_i)$.

We consider two cases for the proof: In Case~1 we assume $\cmax\leq \cbar$, and in Case~2 we assume otherwise, where the number $\cbar$ is the cost at which $f_{r_n}(\cbar)=1-1/e$, more precisely, this happens at $\cbar=r_n (e-e^{1-1/e})$. 

\paragraph{Case 1} In this case, observe that we have $f_{r_n}(c_i) \geq 1-1/e$ for all $i\in S$, which implies $f_{r_i}(c_i) \geq 1-1/e$. This just means $U \geq (1-1/e)n\geq (1-1/e)\uopt$.

\paragraph{Case 2} Let $U_n = \sum_{i\in S} f_{r_n}(c_i)$, we will show that 
\begin{align} 
U_n \geq (1-1/e)\cdot (1-o(1))\cdot \uopt. \label{eq.lun}
\end{align}
 To prove this, consider an auxiliary instance in which, instead of budget $B$, we have a reduced budget $B'=\sum _{i\in S} Q_{r_n}(c_i)$. Note that if we run Mechanism \scale{} on the auxiliary instance, then its stopping rate is $r_n$, and so, the utility gained by the mechanism is exactly $U_n$. 
Let $\uopt_{\aux}$ denote the optimal utility in the auxiliary instance. Then, by applying Lemma \ref{lem.key2approx} on the auxiliary instance, we have $U_n \geq (1-1/e)\cdot \uopt_{\aux}$.
So, if we show that 
\begin{align}
\uopt_{\aux}\geq (1-o(1))\cdot \uopt \label{eq.randomname}
\end{align}
then \eqref{eq.lun} holds and the proof is complete.

We use Lemma~\ref{lem.concave} to prove \eqref{eq.randomname}: First, we show that $B'\geq (1-o(1))\cdot B$; then, applying Lemma~\ref{lem.concave} would imply that  $u^\star(B')\geq (1-o(1)) \cdot u^\star(B)$, which is identical to \eqref{eq.randomname} by definition. So all we need to complete the proof is showing that $B'\geq (1-o(1))\cdot B$.

To this end, we prove that $B'\geq (1-\alpha \cdot \frac{\cmax}{B})\cdot B$, where $\alpha$ is a constant with value $(e-e^{1-1/e})^{-1}\approx 6/5$. This would prove the Lemma due to the large market assumption. First, observe that 
\begin{align}
B&=Q_{r_n}(0)+\sum_{i\in S\backslash\{n\}} Q_{r_n} (c_i) \leq Q_{r_n}(0)+B' \nonumber \\
&\Rightarrow B'\geq B-Q_{r_n}(0) \geq B-r_n.  \label{eq.bplb}
\end{align}
Now, recall that in Case 2, we have $\cmax\geq \cbar$, which implies
\begin{align}
B&\geq \cbar\cdot \frac{B}{\cmax} = r_n (e-e^{1-1/e})\cdot \frac{B}{\cmax} \nonumber\\
&\Rightarrow \cmax \cdot(e-e^{1-1/e})^{-1} \geq r_n. \label{eq.rnub}
\end{align}
Combining \eqref{eq.bplb} and \eqref{eq.rnub} implies $B'\geq (1-\alpha \cdot \frac{\cmax}{B})\cdot B$ with the promised value for $\alpha$.
\end{proof}

\section{Impossibility Result: On why $1-1/e$ is the best approximation possible} \label{sec.hardness}
In this section we show that no truthful (and possibly) randomized mechanism achieves approximation ratio better than $1-1/e$. We prove a stronger claim by allowing satisfying budget feasibility in expectation, i.e. we prove that no truthful mechanism that is budget feasible in expectation can achieve ratio better than $1-1/e$. 
From now on in this section, we assume that all the mechanisms that we refer to are truthful, and are also budget feasible in expectation.
First, we prove the claim assuming that the items are indivisible, then we will see that the same proof easily extends to divisible items as well. 

\paragraph{Proof Outline.} We construct a bayesian instance of the problem and prove that no budget feasible truthful mechanism for this instance can achieve approximation ratio better than $1-1/e$; this also implies that no mechanism for the prior-free setting can achieve ratio better than $1-1/e$~\footnote{This is so because an $\alpha$-approximate mechanism in the prior-free setting is also $\alpha$-approximate in the bayesian setting}.
The proof is done in two steps. First, we show that for any truthful mechanism for this instance, there exist a simple posted price mechanism that achieves at least the same revenue. The posted price mechanism simply offers the same price $p$ to every seller and pays $p$ to any seller who accepts the offer and $0$ to others. In the second step of the proof, we show that for any choice of $p$, such mechanisms can not achieve a ratio better than $1-1/e$. 
The proof that we present w.l.o.g. analyzes the market in expectation: budget feasibility is satisfied in expectation; also, the utility of the mechanisms are computed in expectation. \\

\noindent
We now give the full proof by first giving our hardness instance.


\paragraph{The Hardness Instance.} We construct a bayesian instance of the problem in which all the sellers have unit utility and their costs are drawn i.i.d. from a distribution with CDF $F$, defined as follows: 
\begin{align*}
F(x)=
\begin{cases}
1/e & \textrm{if } x=0, \\
\frac{1}{e(1-x)} & \textrm{if } 0<x<=1-1/e.
\end{cases}
\end{align*}
In other words, $F(x)$ denotes the probability that the cost of a seller is at most $x$. 
Let $\mathcal{D}$ be the distribution defined by $F$ and let $\over{c}$ denote the expected cost of a seller sampled from $\calD$, i.e. $\over{c}=\bbE_{x\sim \calD}[x]$.
We define the budget to be $B=\over{c}\cdot N$ where $N\geq 1$ is an arbitrary integer denoting the number of sellers.

\begin{definition}
A {\em posted price mechanism} is a mechanism that offers a price $p_i$ to any seller $i\in S$, and pays her $p_i$ if she accepts the offer and pays her $0$ otherwise.
\end{definition}
\begin{definition}
A {\em uniform posted price mechanism} is a posted price mechanism that offers the same price to all sellers.
\end{definition}

\begin{definition}
A cut-off allocation rule is an allocation rule which allocates the whole unit of an item if its cost is less than a certain cut-off and allocate $0$ units otherwise.
Let $\cutoff(p)$ denote a cut-off allocation rule with the cut-off price $p$. 
\end{definition}
It is clear that posted price mechanisms use cut-off allocation rules to allocate items from sellers.

\begin{lemma}
If the sellers costs are drawn i.i.d. from the distribution $\calD$, then for any mechanism in this bayesian setting there exists a posted price mechanism with the same approximation ratio. 
\end{lemma}
\begin{proof}
Due to Myerson's Lemma (see footnote \ref{myersonfnote}), any truthful mechanism in the bayesian setting can be seen as a set of allocation and payment rules corresponding to each seller, where the allocation rule is a decreasing function (of the cost) and the payment rule is defined with respect to the allocation rule as we saw in Figure \ref{fig.payment}. Given such an allocation rule for an arbitrary seller $i\in S$, namely $A_i$, one can think of a simpler way to implement $A_i$ (in expectation) by finding a distribution $\pi_i$ over cut-off allocation rules. 

More precisely, we find the distribution $\pi_i$ with PDF $f_i$ such that the distribution $\pi_i$ over all cut-off allocation rules, which assigns probability density $f_i(x)$ to the cut-off allocation rule $\cutoff(x)$, implements the allocation rule $A_i$ in expectation.
 We prove the existence of $\pi_i$ in the following claim.
 
\begin{claim}
Define the distribution $\pi_i$ by its CDF $F_i(\cdot)$ such that $F_i(x)=1-A_i(x)$ for any cost $x\geq 0$. Then $\pi_i$ would implement $A_i$ in expectation. 
\end{claim}
\begin{proof}
All we need to show that a seller with price $x$ would be allocated with probability $A_i(x)$ in $\pi_i$. To see this, note that the probability that the seller is allocated is exactly equal to the probability of observing a cut-off price at least $x$ when a cut-off price is sampled from $\pi_i$. This probability is equal to $1-(1-A_i(x))$ by the definition of $\pi_i$; this proves the claim.
\end{proof}

Now, we claim that the cut-off allocation rule with the cut-off price 
\begin{align}
p_i = F^{-1}\left(\int_0^\infty f_i(p) \cdot F(p) dp \right) \label{eq.pidef}
\end{align}
achieves the same utility and spends the same budget (in expectation) as the allocation rule $A_i$ paired with its corresponding Myerson payment rule.

\begin{claim} \label{clm.samebudget}
For any seller $i\in S$, $\cutoff(p_i)$ achieves the same utility and spends the same budget (in expectation) as the allocation rule $A_i$ paired with its corresponding Myerson payment rule.
\end{claim}
\begin{proof}
The main idea of the proof is that the set of points $P=\{(F(x),xF(x)): 0\leq x < 1\}$ forms a straight line in the two-dimensional plane; See Figure \ref{fig.line} for a proof by picture. 
 \begin{figure}[h!]
\center
\includegraphics[width=0.5\textwidth]{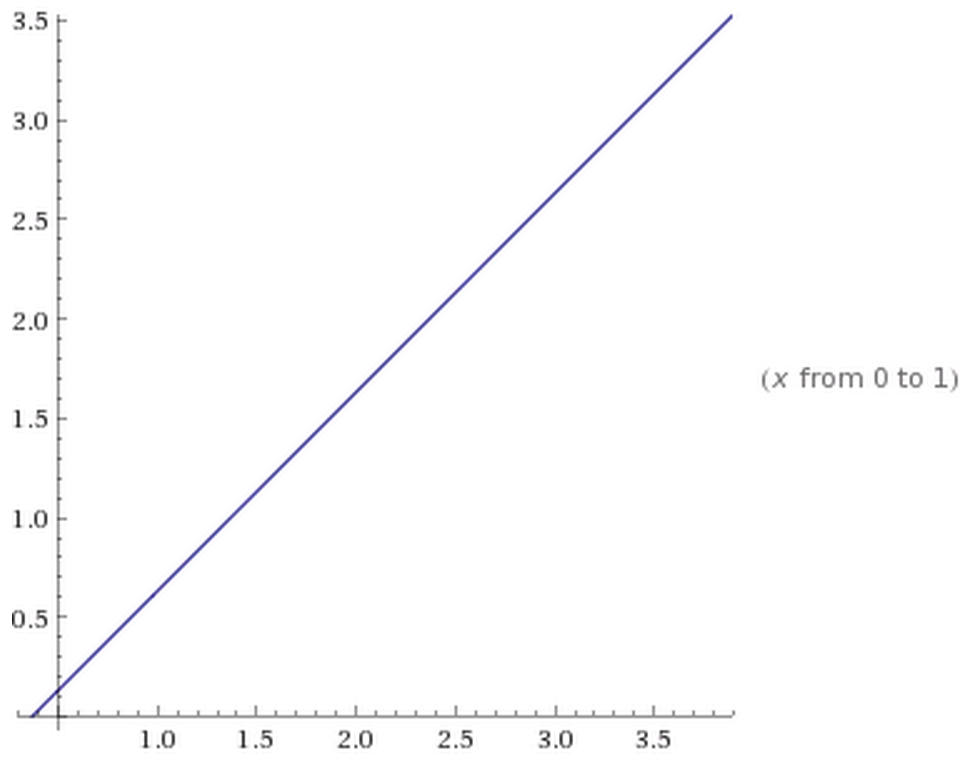}
 \caption{This parametric plot is representing the set of points $P$ for $x=0$ to $x=1$. The horizontal and vertical axis respectively represent $F(x)$ and $xF(x)$}
 \label{fig.line}
 \end{figure}
Note that $F(x)$ denotes the expected allocation when a price $x$ is offered to a seller and $xF(x)$ is the corresponding expected payment. 

Now, see that the expected utility achieved by the allocation rule $A_i$ is $\int_0^\infty f_i(p) F(p) dp$, which is exactly equal to the expected utility achieved by $\cutoff(p)$ due to \eqref{eq.pidef}.  
To prove the claim, it remains to verify that the Myerson payment rules corresponding to $\cutoff(p_i)$ and $A_i$ spend the same budget (in expectation). 
To this end, just observe that $P$ is a straight line; consequently, since the allocation rules $A_i, \cutoff(p_i)$ allocate equal units of items (in expectation), then they also spend the same amount of the budget (in expectation). 
\end{proof}

Due to Claim \ref{clm.samebudget}, the posted price mechanism that offers price $p_i$ to seller $i$ is budget feasible (in expectation) and also achieves an expected utility equal to the utility of the originally given mechanism. This proves the claim.
\end{proof}

\begin{lemma} \label{lem.unipost}
If the sellers costs are drawn i.i.d. from the distribution $\calD$, then for any (budget feasible) posted price mechanism there exists a (budget feasible) uniform posted price mechanism with the same approximation ratio.
\end{lemma}
\begin{proof}
Suppose that $\{p_i\}_{i\in S}$ denotes the offered prices in a posted price mechanism and let 
\begin{align*}
\over{p}=F^{-1}\left(\frac{1}{|S|}\cdot \sum_{i\in S} F(p_i)\right).
 \end{align*}
 First, observe that the uniform posted price mechanism with price $\over{p}$ achieves a utility equal to the utility of the original posted price mechanism; this can be verified simply due to linearity of expectation. 
It remains to verify that the uniform posted price mechanism is budget feasible.
To this end, just observe that the set of points $P$ (depicted in Figure~\ref{fig.line}) is a straight line; consequently, since the posted price mechanism and the uniform posted price mechanism allocate equal units of items (in expectation), then they also spend the same amount of the budget (in expectation). 
\end{proof}

\begin{theorem} \label{thm.hardness1}
For the case of indivisible items, no truthful budget feasible mechanism can achieve approximation ratio better than $1-1/e$.
\end{theorem}
\begin{proof}
We use Lemma \ref{lem.unipost} and show that no uniform posted price mechanism can achieve ratio better than $1-1/e$. Equivalently, we show that the uniform posted price mechanism which spends all the budget in expectation has approximation ratio no better than $1-1/e$.

To define the uniform posted price mechanism that spends all the budget in expectation, we need to find $p^*$ such that $p^* F(p^*)\cdot N=B$. Given the definitions of $F(\cdot)$ and $B$, we can solve this equation to get $p^*=\frac{e-2}{e-1}$. Now, we are ready to compute the approximation ratio. First, note that the (expected) utility of the uniform posted price mechanism is $N\cdot F(p^*)$. If we had $\sum_{i\in S} c_i \leq B$, then we had $U^*=N$ (the optimum solution could buy all items), and so we could write the approximation ratio as
\begin{align*}
\frac{N\cdot F(p^*)}{N} = F(p^*) = 1-1/e,
\end{align*}
which would prove the claim. However, although $\bbE\left[\sum_{i\in S} c_i \right]=B$, the sum is not always bounded by $B$, which means $U^*=N$ does not always hold. We find a way to fix this issue using Hoeffding bounds (see Section \ref{sec.hoeffding} to see formal statements of Hoeffding bounds). We show that although $\sum_{i\in S} c_i$ is not always bounded by $B$, it is concentrated around its mean, $B$, with high probability. We will see that this is enough to prove the theorem.

As a consequence of Hoeffding bounds (stated in Section \ref{sec.hoeffding}), for any $\epsilon > 0$ we have:
\begin{align}
\Pr\left[ \sum_{i\in S} c_i \geq (1+\epsilon)\cdot B \right] \leq e^{-|S|}\label{eq.hoeffding}
\end{align}
Recall that we defined $N=|S|$ and that in our hardness instance $N\to \infty$. Using \eqref{eq.hoeffding}, we will provide an upper bound on the approximation ratio which, for any constant $\epsilon > 0$, approaches to $(1-1/e)(1+\epsilon)$ as $N$ approaches infinity. This proves that the approximation ratio can not be a constant larger than $1-1/e$.

To this end, first note that if $\sum_{i\in S} c_i \leq B (1+\epsilon)$, then we have $U^*\geq \frac{N}{1+\epsilon}-1$; this holds due to Lemma \ref{lem.concave}.
We can use this fact along with \eqref{eq.hoeffding} to write the following upper bound on the (expected) approximation ratio:
\begin{align*}
(1-e^{-N})\cdot \frac{N\cdot F(p^*)}{N/(1+\epsilon)-1} + e^{-N}\cdot 1.
\end{align*}
The above ratio clearly approaches $F(p^*)(1+\epsilon)$ as $N\to \infty$. Noting that $F(p^*)=1-1/e$ finishes the proof.
\end{proof}

Now we use Theorem \ref{thm.hardness1} to prove its counterpart for divisible items.
\begin{corollary}
For the case of divisible items, no truthful budget feasible mechanism can achieve approximation ratio better than $1-1/e$.
\end{corollary}
\begin{proof}
Proof by contradiction. Suppose there exists a mechanism with approximation ratio $\alpha > 1-1/e$ for some constant $\alpha$. Then, we show that we can convert this mechanism to an $\alpha$-approximation mechanism for indivisible items which is truthful and budget feasible in expectation. This would contradict Theorem \ref{thm.hardness1}. 

To do this conversion, we repeat the exact same argument that we used to prove Theorem \ref{thm.hardness1}. As the result, we can convert the given $\alpha$-approximation mechanism to a uniform posted price mechanism with approximation ratio $\alpha$. Note that all posted price mechanisms allocate items without dividing them. Consequently, we have an $\alpha$-approximation mechanism for indivisible items. Contradiction. 
\end{proof}

\section{Conclusion}
Our main contribution is designing optimal budget feasible mechanisms for the knapsack model in large markets. First, we assume that the items are divisible, and study a natural class of deterministic mechanisms: each mechanism in this class is characterized by a decreasing allocation function. All the mechanisms in this class are individually rational, truthful and budget feasible, but they have different approximation ratios based on the choice of the allocation function. We find a mechanism in this class which has an approximation ratio $1-1/e$, and prove that no truthful mechanism can achieve a better approximation ratio.

We also provide a mechanism with approximation ratio $1-1/e$ for the case of indivisible items: the idea is to first run the mechanism for divisible items, and then round the obtained fractional solution (allocation). We design a rounding process that takes the fractional allocation as its input and outputs an integral allocation with its associated payments. Due to the properties of our rounding process, the resulting mechanism is individual rational, truthful-in-expectation, and budget feasible; also, it has approximation ratio $1-1/e$ in large markets.

Finally, we study the problem for submodular utility functions with indivisible items. For this case, we first design a deterministic mechanism which has approximation ratio $\frac{1}{2}$ in large markets; this mechanism can have an exponential running time in general. Inspired by this mechanism, we also design a polynomial-time deterministic mechanism with approximation ratio $\frac{1}{3}$.
We do not provide any results for when the items are divisible in the submodular model: One has to model the utility function over divisible items; the multilinear extension \cite{jan-multilinear} or Lov\`{a}sz extension of submodular functions is a potential choice for this purpose. We leave open this case for future study.

Bayesian
\section*{Acknowledgments} We acknowledge the comments and ideas of an anonymous reviewer that helped us generalize our impossibility result to the case of bayesian setting.

\bibliography{main}

\appendix
\section{Analyzing our Optimal Truthful Mechanism for the general case} \label{sec.generalu}
In this section we will prove that the approximation ratio of Mechanism \ref{alg.truthful} approaches $1-{1}/{e}$ as $\theta$, the market's largeness ratio, approaches $0$. We emphasize that here we dismiss the extra assumption that was made in Section \ref{SEC.TRUTHU1}: There, we assumed all items have utility $1$, here we give a proof for the general case when item $i$ provides utility $u_i > 0$. 

\begin{lemma}
\label{lem.ri}
For each $k\in\{1,\dots, n\}$, $r_k\geq (1-\theta)r^\star$.
\end{lemma}
\begin{proof}
We just need to prove that $f_{(1-\theta)r^\star}$ is not a fit rule (i.e. does not consume all of the budget) when we set the cost of item $k$ to $0$.
First of all, note that 
\begin{align*}
Q_{(1-\theta)r^\star}(x)=(1-\theta)Q_{r^\star}(\frac{x}{1-\theta})\leq (1-\theta)Q_{r^\star}(x)
\end{align*}
Here we used the fact that $Q_{r^\star}$ is a decreasing function.
This implies that $\sum_{i=1}^{n}u_iQ_{(1-\theta)r^\star}(\frac{c_i}{u_i})\leq (1-\theta)B$.
This expression is the budget consumed by the rule $f_{(1-\theta)r^\star}$ without setting the cost of item $k$ to $0$.
When we set $c_k$ to $0$, the amount of budget consumed can be bounded in the following manner
\begin{align}
u_kQ_{(1-\theta)r^\star}(0)+\sum_{i\neq k}u_iQ_{(1-\theta)r^\star}(\frac{u_i}{c_i})\leq (1-\theta)B+u_k\left(Q_{(1-\theta)r^\star}(0)-Q_{(1-\theta)r^\star}(\frac{c_k}{u_k})\right)
\label{eq.thetadiff}
\end{align}
Note that $Q_{(1-\theta)r^\star}(\cdot)$ is defined to be the area of the shaded region as seen in figure \ref{fig.payment}. Therefore one
can crudely upper bound the difference $Q_{(1-\theta)r^\star}(0)-Q_{(1-\theta)r^\star}(x)$ by $x\times f_{(1-\theta)r^\star}(0)$ for any
$x\geq 0$. Now letting $x=\frac{c_k}{u_k}$, and substituting in inequality \ref{eq.thetadiff} we get
\begin{align*}
u_kQ_{(1-\theta)r^\star}(0)+\sum_{i\neq k}u_iQ_{(1-\theta)r^\star}(\frac{u_i}{c_i})&\leq(1-\theta)B+u_k(\frac{c_k}{u_k}-0)\\
&=(1-\theta)B+c_k\leq B
\end{align*}
This completes the proof.
\end{proof}

\begin{lemma}
Mechanism \ref{alg.truthful} has an approximation ratio approaching $1-\frac{1}{e}$ as $\theta$ approaches $0$.
\end{lemma}
\begin{proof}
W.l.o.g. assume that $r^\star=1$ (since we can scale the budget and costs by an appropriate scaling factor). Now let us pick a
constant threshold $0<s<e-1$ and partition the indices $\{1,\dots, n\}$ into two sets $\mathcal{I}$ and $\mathcal{J}$: let $\mathcal{J}$ be
the set of indices $i$ where $\frac{c_i}{u_i}>s$ and let $\mathcal{I}$ be the complement.

Let $r^+$ be the minimum $r_i$ where $i\in \mathcal{J}$. If $\mathcal{J}$ happens to be empty, let $r^+=r^\star=1$. Let $B'$ be
the budget consumed by the allocation rule $f_{r^+}$, i.e. let $B'=\sum_{i=1}^{n}u_iQ_{r^+}(\frac{c_i}{u_i})$. We will prove
that $B'$ is close to $B$. If $r^+=r^\star$, this is obviously true because $B'=B$. So assume that $r^+=r_k$ for some $k\in\mathcal{J}$.

Because of the way $r_k$ is chosen, we have
\begin{align}
B&=u_kQ_{r_k}(0)+\sum_{i\neq k}u_iQ_{r_k}(\frac{c_i}{u_i})
\leq u_k+\sum_{i\neq k}u_iQ_{r_k}(\frac{c_i}{u_i})
\label{eq.bbprime}
\end{align}
Here we used the fact that $Q_{r_k}(0)\leq Q_{r^\star}(0)=1$ (since we assumed $r^\star=1$). Note that $B'\geq \sum_{i\neq k}u_iQ_{r^+}(\frac{c_i}{u_i})$. Combining
this with the inequality \ref{eq.bbprime} we get
\begin{align*}
B'\geq B-u_k=B-c_k\frac{u_k}{c_k}\geq B-\frac{c_k}{s}\geq (1-\frac{\theta}{s})B
\end{align*}

Using lemma $\ref{lem.concave}$, one can see that $u^\star(B')\geq (1-\frac{\theta}{s})u^\star(B)$. But we also know from lemma \ref{lem.key2approx}
that the utility achieved by $f_{r^+}$ is at least $(1-\frac{1}{e})u^\star(B')$. Therefore we have

\begin{align}
\sum_{i=1}^{n}u_if_{r^+}(\frac{c_i}{u_i})\geq (1-\frac{\theta}{s})(1-\frac{1}{e})u^\star(B)
\label{eq.approxutility}
\end{align}

For an item $i\in\mathcal{I}$, we have $r_i\geq (1-\theta)r^\star=1-\theta$ (we used lemma \ref{lem.ri}). Therefore
$$\frac{f_{r_i}(\frac{c_i}{u_i})}{f(\frac{c_i}{u_i})}=\frac{f(\frac{1}{r_i}\frac{c_i}{u_i})}{f(\frac{c_i}{u_i})}\geq \frac{f(\frac{1}{1-\theta}\frac{c_i}{u_i})}{f(\frac{c_i}{u_i})}$$
One can easily verify that $\ln f$ is a concave function. Therefore $\frac{f(\frac{1}{r_i}x)}{f(x)}$ for $x\leq s$ is minimized at $x=s$. This means that
$$\frac{f_{r_i}(\frac{c_i}{u_i})}{f(\frac{c_i}{u_i})}\geq \frac{f(\frac{s}{r_i})}{f(s)}\geq \frac{f(\frac{s}{1-\theta})}{f(s)}$$
If we let $\alpha = \frac{f(\frac{s}{r_i})}{f(s)}$, then for every $i\in\mathcal{I}$ we have
$$f_{r_i}(\frac{c_i}{u_i})\geq \alpha f(\frac{c_i}{u_i})\geq \alpha f_{r^+}(\frac{c_i}{u_i})$$

Similarly, for every item $i\in\mathcal{J}$, $r_i\geq r^+$ and therefore $f_{r_i}(\frac{c_i}{u_i})\geq f_{r^+}(\frac{c_i}{u_i})\geq \alpha f_{r^+}(\frac{c_i}{u_i})$.

We just proved that for every $i\in\{1,\dots, n\}$, $f_{r_i}(\frac{c_i}{u_i})\geq \alpha f_{r^+}(\frac{c_i}{u_i})$. Combining this
with inequality \ref{eq.approxutility} we get
$$\sum_{i=1}^{n}u_if_{r_i}(\frac{c_i}{u_i})\geq \alpha(1-\frac{\theta}{s})(1-\frac{1}{e})u^\star(B)$$

So the approximation ratio for Mechanism \ref{alg.truthful} is at least 
$$\alpha(1-\frac{\theta}{s})(1-\frac{1}{e})=\frac{\ln(e-\frac{s}{1-\theta})}{\ln(e-s)}(1-\frac{\theta}{s})(1-\frac{1}{e})$$

For any fixed $s$, strictly smaller than $e-1$, one can observe that the ratio above approaches $1-\frac{1}{e}$ as $\theta\to 0$.
We will not attempt to optimize the value of $s$ for the sake of brevity.
\end{proof}

\section{Missing Proofs from Section \ref{SEC.TRUTHU1}} \label{sec.defferedproofs}

\begin{proof}[of Lemma \ref{lem.rissorted}]

For any $i,j$ such that $i\leq j$, we prove that $r_i\geq r_j$, this would prove the lemma. The proof is by contradiction, suppose $r_i < r_j$.
First, note that $Q_{r_i} (0)+\sum_{k\in S\backslash\{i\}} Q_{r_i}(c_k)$ represents the sum of payments in \scale{} when the cost of $i$ is set to $0$. Since \scale{} spends all of the budget, then we have:
\begin{align} 
Q_{r_i} (0)+\sum_{k\in S\backslash\{i\}} Q_{r_i}(c_k) = B.\label{eq.ij1}
\end{align}

Also, note that $Q_{r_j} (0)+\sum_{k\in S\backslash\{j\}} Q_{r_i}(c_k)$ represents the sum of payments in \scale{} when the cost of $j$ is set to $0$; a similar argument shows that
\begin{align} 
Q_{r_j} (0)+\sum_{k\in S\backslash\{j\}} Q_{r_j}(c_k) = B. \label{eq.ij2}
\end{align}
Taking the difference between \eqref{eq.ij1} and \eqref{eq.ij2} implies that
\begin{align} 
\Big(Q_{r_i}(0) - Q_{r_j} (0) \Big)+ \Big(Q_{r_i}(c_j) - Q_{r_j}(c_i)\Big)= 0 \label{eq.ij3}
\end{align}
Now, observe that since $r_i < r_j$ and $c_i \leq c_j$, then we have that $Q_{r_i}(0) < Q_{r_j}(0)$ and $Q_{r_i}(c_j) < Q_{r_j}(c_i) $. This contradicts with \eqref{eq.ij3}.
\end{proof}

\begin{proof}[of Lemma \ref{lem.concave}]
We need to prove that for $B_1, B_2\geq 0$ and $0\leq \lambda\leq 1$ 
$$\lambda u^\star(B_1)+(1-\lambda)u^\star(B_2)\leq u^\star(\lambda B_1+(1-\lambda)B_2)$$

Let $x_i$ be the amount of item $i$ we allocate to achieve $u^\star(B_1)$ and let $y_i$ be the amount of item $i$ we allocate
to achieve $u^\star(B_2)$.

Now let $z_i=\lambda x_i+(1-\lambda) y_i$. Note that since $0\leq x_i,y_i\leq 1$, we also have $0\leq z_i\leq 1$.
If we allocate $z_i$ from item $i$, the utility we get will be
$$\sum_{i=1}^{n}u_iz_i=\lambda(\sum_{i=1}^{n}u_ix_i)+(1-\lambda)(\sum_{i=1}^{n}u_iy_i)=\lambda u^\star(B_1)+(1-\lambda)u^\star(B_2)$$

The cost paid by these allocations is simply
$$\sum_{i=1}^{n}c_iz_i=\lambda(\sum_{i=1}^{n}c_ix_i)+(1-\lambda)(\sum_{i=1}^{n}c_i y_i)\leq \lambda B_1+(1-\lambda)B_2$$

Therefore $z_i$'s are an allocation that spend a budget of at most $\lambda B_1+(1-\lambda)B_2$ and yet achieve a utility of
$\lambda u^\star(B_1)+(1-\lambda)u^\star(B_2)$. This proves that
$$u^\star(\lambda B_1+(1-\lambda)B_2)\geq \lambda u^\star(B_1)+(1-\lambda)u^\star(B_2)$$
\end{proof}

\section{Mechanisms for Indivisible Items} \label{sec.rounding}

In this section, using our mechanism for divisible items, we design a mechanism for indivisible items with approximation ratio $1-1/e$.
The idea is to first run the mechanism for divisible items, and then round the obtained fractional solution (allocation). We design a rounding process that takes the fractional allocation as its input and outputs an integral allocation with its associated payments. Due to the properties of our rounding process, the resulting mechanism is individual rational, truthful, and budget feasible; also, it has approximation ratio $1-1/e$ in large markets. 

First, we explain a set of properties that we need the rounding procedure to satisfy. If the rounding procedure satisfies these properties, then its individual rationality, truthfulness, and budget feasibility would be guaranteed. Also, these properties guarantee that the approximation ratio would remain $1-1/e$. First we explain these properties in Section \ref{sec.roundingprop}, then, we state our rounding procedure and prove that it satisfies these desired properties in Section \ref{sec.mainrounding}.

\subsection{Properties of the Rounding Procedure} \label{sec.roundingprop}
Let $\tix_1,\ldots,\tix_n$ represent a fractional allocation where $\tix_i$ denotes the allocated fraction from seller $i\in S$; also, let $\tip_1,\ldots,\tip_n$ be the associated payments for this allocation. We round this solution to an integral solution, represented by the allocation $x_1,\ldots,x_n$ and payments $p_1,\ldots,p_n$, such that:
\begin{enumerate}
\item \label{e.p1} Item $i$ is bought with probability $\tix_i$.
\item \label{e.p2} If item $i$ is bought, then $p_i = \tip_i/\tix_i$, and $p_i = 0$ otherwise.
\item \label{e.p3} $\sum_{i\in S} p_i \leq B+\cmax$.
\end{enumerate}

Properties \ref{e.p1} and \ref{e.p2} imply individual rationality and truthfulness: 
Verifying individual rationality is straight-forward due to the individual rationality of the fractional solution. 
For truthfulness, just see that $\bbE[x_i]=\tix_i$ and $\bbE[p_i]=\tip_i$, which implies $\bbE[p_i - c_i x_i] = \tip_i - c_i \tix_i$. 
This just means that, in the mechanism for indivisible items, $i$ cannot benefit (in expectation) by misreporting, since she is already receiving the maximum possible expected utility that she can ever achieve. 

Properties \ref{e.p1} and \ref{e.p2} also imply budget feasibility in expectation, however, Property \ref{e.p3} provides a much stronger guarantee: the budget will not be violated by an additive factor more than $\cmax$.\footnote{We can always reduce the budget slightly to get strict budget feasiblility, e.g. we can reduce the budget to $(1-\epsilon)B$ for an arbitrary small $\epsilon>0$. This will not affect the approximation ratio (asymptotically) in a large market.}

\subsection{Description of the Rounding Procedure} \label{sec.mainrounding}
In this section, we focus in designing a rounding procedure which satisfies Properties \ref{e.p1}, \ref{e.p2} and \ref{e.p3}. To this end, we first need to define a polytope $\calP$ that represents all the (fractional) allocations which, in a certain sense, are budget feasible:
\begin{align*}
\calP = \left\{ y\in [0,1]^n:  \sum_{i\in S} y_i \cdot \frac{\tip_i}{\tix_i} \leq B \right\}
\end{align*}

First, we prove that extreme points of $\calP$ are ``almost'' integral.
\begin{definition}
A point $y\in [0,1]^n$ is called {\em semi-integral} if there is at most one entry of $y$ which is non-integral, i.e. there is at most one index $i$ such that $0< y_i < 1$.
\end{definition}

\begin{lemma}
All the extreme points of $\calP$ are semi-integral.
\end{lemma}
\begin{proof}
The proof is straight-forward, we give a high-level description and omit the formal details.
The idea is to see $\calP$ as the intersection of a hypercube and a hyperplane; the hypercube is $[0,1]^n$ and the hyperplane is $\sum_{i\in S} y_i \cdot \frac{\tip_i}{\tix_i} \leq B$.
 So, any extreme point of $\calP$ either is an extreme point of the hypercube (which is integral), or is on the intersection of the hyperplane and an edge of the hypercube. In the latter case, it can be seen that such a point has at most one fractional entry, since any two adjacent vertices on the hypercube are different in at most one entry.
 \end{proof}

\paragraph{Outline of the Rounding Procedure}
The procedure accepts the fractional allocation constructed by the mechanism, i.e. $\tix=(\tix_1,\ldots,\tix_n)$, and then writes it as a convex combination of extreme points of $\calP$. Then, it samples an extreme point from the convex combination, where each point is selected with probability proportional to its coefficients in the convex combination. Finally, it rounds the sampled extreme point (which is a semi-integral point) to an integral point. 
We use the following fact about semi-integral points for implementing the last step:

\begin{fact}
A semi-integral point $y\in [0,1]^n$ can be written as the convex combination of two integral points which differ in at most one entry, i.e. $y=\alpha y' + (1-\alpha) y''$ where $y',y'' \in \{0,1\}^n$ are integral points which differ in at most one entry.
\end{fact}

Now we are ready to formally state our main rounding procedure.
\begin{procedure}
\SetKwInOut{Input}{input}
\SetKwInOut{Output}{output}
\Input{Allocation vector $\tix$ and Payment vector $\tip$}
\BlankLine
Find extreme points $z^1,\ldots, z^K \in \calP$ and positive numbers $\lambda_1,\ldots,\lambda_K$ summing up to one such that $\tix = \sum_{i=1}^K \lambda_i \cdot z^i$ \;
Sample a single point from $\{z^1,\ldots, z^K\}$ where $z^i$ is selected with probability $\lambda_i$; Let $z$ denote the sampled point\;
Write $z$ as the convex combination of two integral points $x^1, x^2$ such that $x_1,x_2$ differ in at most one entry, i.e. suppose $z=\alpha\cdot x^1 + (1-\alpha) x^2$\;
With probability $\alpha$, let $x=x^1$, otherwise, let $x=x^2$\;
Announce $x$ as the final allocation and pay $x_i \cdot \tip_i/\tix_i$ to seller $i$.
\caption{EfficientRounding() {\label{alg.rounding2}}}
\end{procedure}

\begin{lemma}
Procedure \ref{alg.rounding2} satisfies Properties \ref{e.p1}, \ref{e.p2} and \ref{e.p3}.
\end{lemma}
\begin{proof}
It is straight-forward to verify that Properties  \ref{e.p1} and \ref{e.p2} hold; for any seller $i\in S$ we have:
\begin{align*}
\bbE [x_i] =& \bbE[\alpha\cdot x^1_i + (1-\alpha)\cdot x^2_i] \\
=& \bbE[z_i] = \bbE \sum_{j=1}^K \lambda_i \cdot z^j_i = \tix_i
\end{align*}
which implies Property \ref{e.p1} since $x_i$ is a binary random variable. Property \ref{e.p2} trivially holds by the construction of Procedure \ref{alg.rounding2}. 

It remains to prove Property \ref{e.p3}. To this end, define $p(y)=  \sum_{i\in S} y_i \cdot {\tip_i}/{\tix_i}$ for any $y\in [0,1]^n$. We prove the claim by showing that $p(x)\leq B+\cmax$. 
Equivalently, we can show that
$p(x^1) \leq B+\cmax$ and
$p(x^2) \leq B+\cmax$. We prove this only for $x^1$, the proof for $x^2$ is identical. The claim is trivial if $x^1=x^2$, since in this case we have $x^1\in \calP$, which means $p(x^1) \leq B$. So suppose $x^1\neq x^2$.
Recall that we have
\begin{align*}
z=\alpha\cdot x^1 + (1-\alpha) x^2,
\end{align*}
where $x^1,x^2$ are two adjacent vertices on the hypercube. Also, recall that any two adjacent vertices on the hypercube are different in exactly one entry, so suppose $x^1,x^2$ are different in entry $j$, i.e. $x^1_j \neq x^2_j$.
Now, we prove the lemma by showing that
\begin{align}
p(x_1)\leq p(z) + \tip_j\cdot\frac{1-\tix_j}{\tix_j}  \leq B + \cmax.\label{eq.key2round}
\end{align}

First verify that the first inequality in \eqref{eq.key2round} holds since $x_1$ and $z$ are only different in their $j$-th entry: if $x^1_j=0$, then $p(x_1)\leq p(z)$; if $x^1_j=1$, then it is straight-forward to verify that $p(x_1) = p(z) + \tip_j(1-\tix_j)/\tix_j$. 
Having that the first inequality holds, \eqref{eq.key2round} is proved if we show that 
\begin{align}
p(z)& \leq B,\label{eq.rkey2}\\
\tip_j(1-\tix_j)/\tix_j &\leq \cmax.  \label{eq.rkey3}
\end{align}
To verify \eqref{eq.rkey2}, just note that $z\in \calP$.
To verify \eqref{eq.rkey3}, recall that $\tix_j=f_r(c_j/u_j)$ and $\tip_j = u_j\cdot Q_{r}(c_j/u_j)$ for some $r>0$; we will prove the following bounds on $r$: 
\begin{align}
\tip_j/\tix_j &\leq  r u_j, \label{eq.key11}\\
r  &\leq  \frac{c_j/u_j}{1-\tix_j} \label{eq.key22}.
\end{align}
Observe that combining \eqref{eq.key11} and \eqref{eq.key22} implies \eqref{eq.rkey3}. So, we are done if we prove \eqref{eq.key11} and \eqref{eq.key22} hold.
To prove \eqref{eq.key11}, it is enough to note that $Q_{r}(c_j/u_j)\leq r\tix_j$, i.e. the area under the curve $f_r$ that represents the payment per unit of utility to seller $j$ fits in a rectangle with width $r$ and height $\tix_j$; this implies $\tip_j/\tix_j \leq r u_j$. It is also straight-forward to verify \eqref{eq.key22} holds
due to the concavity of $f_r$.
\end{proof}

\section{The Optimal Standard Allocation Rule} \label{sec.lb}
Mechanism \ref{alg.truthful} is defined uniquely by the standard allocation rule $f$. Here we provide an alternative proof for showing that our choice of $f(x)=\ln(e-x)$ is optimal, in the sense that Mechanism \ref{alg.truthful} attains the best approximation ratio under this choice. Although this fact is also a consequence of what we already proved in Section \ref{sec.hardness}, we provide a more direct proof here with a prior-free hardness instance.

For simpler analysis, we first ignore the trick that we used for making our mechanism truthful: we work with the non-truthful, but cleaner mechanism \scale{} and we prove that no approximation
ratio better than $1-\frac{1}{e}$ is attainable using any standard allocation rule $f$. In the end, we see that the trick that makes mechanisms truthful only worsens the approximation ratio.
Therefore the optimality result also applies to the family of truthful mechanisms defined by Mechanism \ref{alg.truthful}.

We provide hardness examples with unit utilities, i.e. with $u_i=1$ for all $i\in S$. We also drop the indices from payment functions $P$. This is done because all the payment functions are identical: Note that \scale{} uses the same allocation rule for all the sellers, which means the unit-payment functions are identical. Also, since all the utilities are $1$, then the unit-payment functions are identical to the payment functions, i.e. $P(c)=Q(c)$.
Unless specified otherwise we also assume that the scaling ratio $r^\star$ is $1$. i.e. we have $P(c)=Q(c)=Q_{1}(c)$ and $f(c)=f_1(c)$.

\begin{definition}
Assume that we are given a standard allocation rule $f$. Define the set $S_f\subseteq \mathbb{R}^2$ as follows
$$S_f=\{(f(c), P(c)-c)\mid c\in\mathbb{R}\}$$
where $P(c)$ is the payment rule associated with the allocation rule $f$ (i.e. $P(c)=Q_1(c)$. Also define the set $T_f$ as the downward closure of $S_f$, i.e.
$$T_f=\{(x, y)\in\mathbb{R}^2\mid \exists y'\geq y: (x,y')\in S_f\}$$
\end{definition}

Note that in our definition we simply work with the standard function $f$ and not its scaled variants $f_r$. The following lemma gives us a way to
find out if the approximation ratio attained by the mechanism is worse than a number $\beta$.

\begin{lemma}
\label{opt-lem}
Suppose that the point $(\beta, 0)$ lies in the convex hull of $T_f$. Then the mechanism defined by $f$ has approximation ratio
at most $\beta$.
\end{lemma}
\begin{proof}
Assume that the point $(\beta, 0)$ can be written as the convex combination of points in $T_f$ in the following way
$$(\beta, 0)=\sum_{i=1}^{n}\alpha_i p_i$$
where $\alpha_i\geq 0$, $\sum_{i=1}^{n}\alpha_i=1$, and $p_1,\dots,p_n\in T_f$. Because of the way $T_f$ is defined, for each $p_i$ one can find a value $c_i$ such that
$p_i=(f(c_i),y_i)$ and $y_i\leq P(c_i)-c_i$. Then this means that
\begin{align}
\sum_{i=1}^{n}\alpha_i f(c_i)=&\beta \label{beta-eq}\\
\sum_{i=1}^{n}\alpha_i (P(c_i)-c_i)\geq& 0 \label{price-eq}
\end{align}

Now let $M$ be a very large number and consider $M$ sellers. For each $i$, let $\alpha_i$ fraction of the sellers price their item at $c_i$. The issue of $\alpha_i M$ not being integral can be easily dealt with, but would unnecessarily complicate the proof, hence we assume $\alpha_i M$ is integral.
Assume that each seller's item is worth one unit to the buyer. Now let the buyer's budget $B$ be
$$B=M\sum_{i=1}^{n}\alpha_i P(c_i)$$

With this definition, we are simply saying that $f$ is fit with respect to our budget, or in other words, the mechanism consumes all of the budget at $r=1$.

Because of \eqref{price-eq}, our budget $B$ is at least as large as the sum of the costs of all items $M\sum_{i=1}^{n}\alpha_ic_i$. Therefore we could buy all of the items and obtain a utility of $M$, if we were not constrained by the mechanism. But because of \eqref{beta-eq}, 
our mechanism obtains a utility of $\beta M$.
Therefore the approximation ratio in this example is $\frac{\beta M}{M}=\beta$.
\end{proof}

It is worth mentioning that the reverse of lemma \ref{opt-lem} is also true under some mild conditions, but in this section we are only concerned with the direction proved.

Building on top of our lemma, we can now prove that $1-\frac{1}{e}$ is the best approximation factor among the mechanisms defined by a standard allocation rule.

\begin{theorem} \label{thm.schardnss}
Given any standard allocation rule $f$, the mechanism \scale{} has approximation ratio at most $1-\frac{1}{e}$.
\end{theorem}
\begin{proof}
To prove this, we simply need to show that the point $(1-\frac{1}{e}+\epsilon, 0)$ is inside the convex hull of $T_f$ for every $\epsilon>0$.

Suppose the contrary is true. So for some $\epsilon>0$, the point $(\beta, 0)=(1-\frac{1}{e}+\epsilon, 0)$ is not inside the convex hull of $T_f$. Since the point is not
inside the convex hull, there is a line that separates them. We can further assume that this line passes through the point and is therefore given by the equation $y=s(x-\beta)$ for some slope $s$. Since $T_f$ is downward closed with respect
to the $y$ coordinate, $T_f$ must fall completely below this line. Since $S_f\subseteq T_f$, this means that $S_f$ must also fall completely below this line which means that for all $c$, we have
\begin{equation}
\label{imp-ineq}
P(c)-c\leq s(f(c)-\beta)
\end{equation}

We can safely assume that $s>0$. If this was not true, then for $c=0$, we would necessarily have $f(c)<\beta$. But this already shows that the mechanism never buys more than a fraction of $\beta$ of any item, and therefore the approximation factor cannot be larger than $\beta$.

Now consider the function $f_0(x)=\ln(\max(e-tx, 1))$ for some given $t>0$. For this alloaction function let $P_0$ be the Myerson's payment rule, i.e. $P_0(x)=xf_0(x)+\int_{x}^{\infty}f(x)dx$ (the shaded area in figure \ref{fig.payment}). One can easily verify that for $0\leq c\leq \frac{e-1}{t}$
$$\int_{c}^{\infty}f_0(x)dx=c-\frac{e-1}{t}-(c-\frac{e}{t})\ln(e-tc)$$
and therefore
\begin{equation}
\label{imp-eq}
P_0(c)-c=\frac{e}{t}\ln(e-tc)-\frac{e-1}{t}=\frac{e}{t}(f_0(c)-(1-\frac{1}{e}))
\end{equation}

If we choose $t$ so that $s=\frac{e}{t}$, then equation \ref{imp-eq} becomes very similar to the inequality \ref{imp-ineq}. If we subtract equation \ref{imp-eq} from the inequality \ref{imp-ineq}, then we get
\begin{equation}
\label{tres-ineq}
P(c)-P_0(c)\leq s(f(c)-f_0(c)-\epsilon)
\end{equation}

This inequality holds for $0\leq c\leq \frac{e-1}{t}=s(1-\frac{1}{e})$. For $c=0$, we get that $P(0)-P_0(0)\leq s(f(0)-1-\epsilon)<0$. Therefore $P(0)<P_0(0)$.

Let $c^*=\sup\{c\mid f(c)< f_0(c)\}$. This supremum is strictly greater than $0$, because otherwise we would have $f(x)\geq f_0(x)$ for all $x>0$, which would mean that $P(0)\geq P_0(0)$, which is a contradiction. This supremum is also finite because for $c\geq \frac{e-1}{t}$ we have $f_0(c)=0$ and $f(c)\geq 0$.

Because $c^*$ is the supremum, we can find a sequence of points $c_1, c_2, \dots$, such that $f(c_i)<f_0(c_i)$ for all $i$, and $c_i\to c^*$. For each such $c_i$, we have $c_i\leq c^*\leq \frac{e-1}{t}$, and therefore the inequality \ref{tres-ineq} holds. This means that
$$c_i(f(c_i)-f_0(c_i))+\int_{c_i}^{\infty}(f(x)-f_0(x))dx\leq s(f(c_i)-f_0(c_i)-\epsilon)$$

By rearranging the terms we get
$$\int_{c_i}^{\infty}(f(x)-f_0(x))dx\leq (s-c_i)(f(c_i)-f_0(c_i))-s\epsilon\leq -s\epsilon$$

Here we used the fact that $c_i\leq c^*\leq s(1-\frac{1}{e})<s$ and $f(c_i)-f_0(c_i)<0$.

But now as we take the limit as $i\to \infty$, we get
$$\int_{c^*}^{\infty}(f(x)-f_0(x))dx\leq -s\epsilon$$
which is a contradiction because for $x>c^*$, we have $f(x)-f_0(x)\geq 0$, and so the integral cannot be negative.
\end{proof}

\begin{theorem}
  For any standard allocation rule $f$, Mechanism \ref{alg.truthful} has approximation ratio at most $1-\frac{1}{e}$.
\end{theorem}
\begin{proof}
  Mechanism \ref{alg.truthful} always achieves a worse or equal utility compared to \scale. Therefore its approximation ratio is at most $1-\frac{1}{e}$ by Theorem \ref{thm.schardnss}.
\end{proof}

\section{Submodular Utility Functions} \label{sec.submod}

In this section, we present truthful mechanisms for the knapsack problem when the utility function for the buyer is a {\em monotone} submodular function rather than an additive function. More precisely, a monotone submodular function $F:2^S\rightarrow \mathbb{R}^+$ defines the utility that the buyer derives from buying a subset of $S$. The buyers problem then becomes selecting a subset $\sstar\subset S$ such that $\sstar$ is budget feasible, i.e. $c(\sstar)\leq B$, and $\sstar$ has the highest utility, $F(\sstar)$, among all the budget feasible subsets. (recall that $c(S^\star)$ denotes $\sum_{i\in S^\star}c_i$)
 
This problem has first been studied in \cite{2010-focs_singer_budget-feasible-mechanisms} for arbitrary markets and a $0.0089$-approximation is presented for it. Later, \cite{2011-soda_improved_budget_feasible} improved this result by giving an exponential-time deterministic mechanism with approximation ratio $0.119$ and a polynomial-time randomized mechanism with approximation ratio $0.126$.

\paragraph{Our Results}
In this section, we focus on studying this problem in {\em large markets} where each individual can not significantly affect the market (see Section \ref{sec.submodprel} for a formal definition). Under this assumption, we design a deterministic mechanism which has approximation ratio $\frac{1}{2}$, but has an exponential running time. 
Later, we will see that the exponential running time is solely due to the computational difficulty of solving the knapsack problem for submodular functions. 
In fact, our mechanism is also a polynomial-time $(\gamma^2/2)$-approximation when it has access to a $\gamma$-approximation oracle for solving the knapsack problem (see Section \ref{sec.exptime}). To the extent of our knowledge, the best existing approximation oracle has $\gamma=1-1/e$ due to \cite{sviridenko}; this provides us a mechanism with approximation ratio $\gamma^2/2\approx 0.2$. 

We take a step further and improve this result by presenting a deterministic polynomial-time $\frac{1}{3}$-approximation mechanism in Section \ref{sec.polytime}. 
This mechanism, although using a greedy optimization oracle with $\gamma=1-1/e$, has approximation ratio equal to $\frac{1}{3}$ (rather than $\gamma^2/2\approx 0.2$).   

\paragraph{Oneway-Truthfulness}
All of the mechanisms that we have are truthful, however, we first present a simpler version of them which are not fully truthful but satisfy truthfulness in a weaker form, which we call {\em oneway-truthfulness}. Briefly, by this property, players only have incentive to report costs lower than their true cost.
 This notion is formally defined in Section \ref{sec.onewaydef}.
 \\
 In Section \ref{sec.o2t}, we convert our oneway-truthful mechanisms to (fully) truthful mechanisms only by changing the payment rule.
It is worth pointing out that analyzing the performance ratio of oneway-truthful mechanisms is not much different than analyzing truthful mechanisms: Since the cost of optimum solution may only decrease if players report lower costs, then any $\alpha$-approximate solution for the reported instance is also an $\alpha$-approximate solution for the original instance.

\subsection{Preliminaries} \label{sec.submodprel}
In this section, first we state a a few basic definitions. Then, we formally define the large market assumption and oneway-truthfulness. Finally, we state a few definitions regarding submodular functions which are used in our mechanisms. 

\subsection{Basic Definitions}
Similar to before, we say a subset of sellers $T\subseteq S$ is budget feasible if $c(T)\leq B$. Utility of the subset $T$ is defined by $F(T)$.
The {\em optimum subset}, $\sstar$, is the budget feasible subset with the highest utility. We call $F(\sstar)$ the {\em optimum utility} and also denote it by $F^\star$.   

\subsubsection{The Large Market Assumption}
Our large market assumption here is almost identical to the alternative large market assumption that was discussed in Section \ref{sec.linearlargemarket}. Intuitively, it says that no individual affects the (optimum solution of the) market significantly. This assumption is formally defined below.

Let $\umax =\max_{s\in S}{F(\{s\})}$ and $U^*$ be the total utility of the optimum solution (i.e. the maximum achievable utility when the costs are known). This large market assumption states that
\begin{definition}
We say that a market is {\em large} if $\umax \ll U^*$.
\end{definition}
In other words, we define the largeness ratio of the market to be $\theta=\frac{\umax}{U^\star}$ and analyze our mechanisms for when $\theta\to 0$.



\subsubsection{Oneway-Truthfulness} \label{sec.onewaydef}
Think of a reverse auction with a set of sellers $S$ where each seller $i\in S$ has a private cost $c_i$. 
In a truthful mechanism, no seller wants to report a fake cost regardless of what others do. In a oneway-truthful mechanism, no seller wants to report a cost higher than its true cost regardless of what others do. 

For clarification, we first define the notion of cost vector briefly: when we say a cost vector $d$, we mean a vector which has an entry $d_i$ corresponding to any seller $i$, where $d_i$ represents the cost associated with seller $i$.
Now we formally define the notion of oneway-truthfulness as follows:
\begin{definition}
A mechanism $\mathcal{M}$ is {\em oneway-truthful} if, for any seller $i\in S$ and any cost vector $d$ for which $d_i>c_i$, we have:
\begin{align*}
u_i({c_i}, d_{-i}) \geq u_i(d)
\end{align*} 
where $d_{-i}$ denotes any cost vector corresponding to the rest of players except $i$ and $u_i(\cdot)$ denotes the utility of player $i$.
\end{definition} 

\subsubsection{Submodular Functions}
Given the submodular function $F:2^S\rightarrow \mathbb{R}^+$, we define an ordering of the elements of $S$ with respect to $F$, which we call the {\em greedy sequence} and denote it by $\chi(F)=\langle x_1,\ldots,x_n\rangle$. For simplicity in the definition, we first define an auxiliary notion as follows: let $\chi_i=\cup_{j=1}^i \{x_i\}$ for all $i$, and let $\chi_0=\emptyset$. 
The sequence is constructed such that 
\begin{align*}
x_i = \argmax_{s\in S\backslash{\chi_{i-1}}} F(\chi_{i-1}\cup\{s\})
\end{align*}
for all positive $i\leq n$. 

It is easy to verify that $\chi(F)$ can be constructed in polynomial time by finding the values of $x_1,\ldots,x_n$ one by one in the order that they appear in $\chi(F)$.
After constructing of the greedy sequence, define $\partial_i=F(\chi_i)-F(\chi_{i-1})$ for all $i\leq n$.

\subsection{The Exp-Time Mechanism} \label{sec.exptime}
In this Section, we present an extremely simple mechanism which we call the {\em Oracle} Mechanism.
Given a submodular function $F:2^{S}\rightarrow \mathbb{R}^+$, the mechanism first finds the optimal budget feasible subset, i.e. the subset $S^\star\subseteq S$ such that $S^\star$ is budget feasible and has the highest utility among all the budget feasible subsets. Let $F^{\star}=F(\sstar)$ and $\rstar=\frac{B}{F^\star}$. We also call $F^\star, \rstar$ respectively the {\em optimum utility} and the {\em optimum cost per utility rate}. We also call $\rstar$ the {\em optimum rate} when there is no risk of confusion.

\paragraph{Winner Selection}
The Mechanism constructs the sequence $\chi(F)$ and chooses the largest integer $k$ such that $F\left(\chi_k\right) \leq F^\star/2$. Then, it reports $\chi_k$ as the set of winners. 

\paragraph{The Payment Rule} 
For simplicity, assume that the winners are indexed from $1,\ldots,k$.
The Payment to winner $i$ is equal to $2 r_i \cdot \partial_i$, where $r_i$ is the optimum cost per utility rate for the instance in which seller $i$ is removed from the set of sellers, $S$. 
In other words, think of an auxiliary instance in which we are given budget $B$ and the cost of every seller is the same as the original instance except that $c_i=\infty$. Then, $r_i$ is the optimum cost per utility rate in this instance.

\begin{algorithm}
\LinesNotNumbered
\SetKwInOut{Input}{input}
\SetKwInOut{Output}{output}
\Input{Submodular utility function $F$, Budget $B$}
\BlankLine
Sellers report their costs\;
Compute the optimum utility $F^\star$ and the optimum cost per utility rate $\rstar$\;
Construct the sequence $\chi(F)$ \;
Find the largest integer $k$ such that $F(\chi_k)\leq F^\star/2$\;
Announce $\chi_k$ as the set of winners\;
\ForEach{$i\in \chi_k$}{
	$S\leftarrow S\backslash \{i\}$\;
	Let $r_i$ be the optimum cost per utility rate in the current instance\;
	Pay $2 r_i\cdot \partial_i$ to seller $i$\;
	$S\leftarrow S\cup \{i\}$\;
}
\caption{Oracle Mechanism{\label{alg.oracle}}}
\end{algorithm}

Below we prove that the Oracle Mechanism is individually rational, oneway-truthful, and it has approximation ratio $\frac{1}{2}$ in large markets. Also, it is {\em almost budget feasible}, i.e. we can show that the total sum of its payments is at most $B+o(B)$. So, all we need for having a strictly budget feasible mechanism is starting with a slightly decreased budget. We will see that this does not affect the approximation ratio of the mechanism asymptotically due to the large market assumption. The formal proof for this is deferred to Section \ref{sec.strictbf}.

\paragraph{Simplifying Assumption} Through out the analysis, w.l.o.g. we assume that the sellers appear in the greedy sequence in an increasing order, i.e. $x_i=i$ for all $i\in S$. 

\begin{lemma}
The Oracle Mechanism is individually rational. 
\end{lemma}
\begin{proof}
We show that $c_{i}\leq 2 \rstar\cdot  \partial_i$ and $\rstar \leq r_i$. These two imply $c_{i}\leq 2 r_i \cdot \partial_i$ which is individual rationality. 
First we prove $\rstar \leq r_i$ as follows. Let $F^\star_i$ denote the optimum utility when seller $i$ is removed. Clearly, we have $F^\star_i\leq F^\star$. This just means $\rstar \leq r_i$ since we have $r_i F^\star_i=\rstar F^\star = B$.

It remains to show that $c_{i}\leq 2 \rstar\cdot  \partial_i$. Note that we only need to show this for the last winner, i.e. it suffices to prove that $c_{k}\leq 2 \rstar\cdot  \partial_k$ since we have $\frac{c_i}{\partial_i}\leq \frac{c_j}{\partial_j}$ iff $i\leq j$. The proof is by contradiction, suppose $c_{k} > 2 \rstar\cdot  \partial_k$. 

 For more intuition, we first explain our argument for contradiction in words and then we state it more formally. 
From the definition of the greedy sequence $\chi(F)$, it can be seen that conditioned on buying the subset $\chi_{k}$, the cost for buying each extra unit of utility is at least $\frac{c_{k}}{\partial_{k}}$, which is more than $2r^\star$. 
So, even if we get the subset $\chi_{k}$ for free, the cost for buying an additional $F({\sstar})-F({\chi_k})$ units of utility (which is needed for the optimum solution) would be 
\begin{align*}
\left( F({\sstar})-F({\chi_k})\right) \cdot \frac{c_{k}}{\partial_{k}} \geq &
 \frac{F^\star}{2}\cdot \frac{c_{k}}{\partial_{k}} \\
 >  &
 \frac{F^\star}{2}\cdot 2\rstar=B.
\end{align*}
This means cost of the optimum solution is more than $B$.

To formalize this contradiction, just note that by monotonicity of $F$, we have
\begin{align}
F({\sstar}\cup \chi_{k})-F({\chi_{k}}) \geq \frac{F^\star}{2}. \label{eq.lbscupchi0}
\end{align}
Now observe that by the definition of the greedy sequence $\chi(F)$, 
the cost for buying each extra unit of utility conditioned on having ${\chi_{k}}$ is at least $\frac{c_{k}}{\partial_{k}}$. This fact, and \eqref{eq.lbscupchi0} together imply that 
\begin{align}
c (\sstar\cup \chi_{k}) - c (\chi_{k}) \geq 
\frac{c_{k}}{\partial_{k}}\cdot \frac{ F^\star}{2}. \label{eq.120}
\end{align}
On the other hand, recall that $c_{k} > 2 \rstar\cdot  \partial_k$, so we can write \eqref{eq.120} as
\begin{align*}
c (\sstar\cup \chi_{k}) - c (\chi_{k}) \geq\,  & 
\frac{c_{k}}{\partial_{k}}\cdot \frac{ F^\star}{2} > \rstar \cdot F^\star = B.
\end{align*}
This implies $c (\sstar) > B$ which is a contradiction with the budget feasibility of $\sstar$.
\end{proof}

\begin{lemma} \label{lem.ortruthful}
The Oracle Mechanism is oneway-truthful. 
\end{lemma}
\begin{proof}
For contradiction, suppose there exists a seller $i\in S$ who has incentive to report a cost $\over{c_i}$ which is higher that its true cost $c_i$. 
Assuming that the mechanism picked $\chi_k$ as the set of winners, we have either $i>k$ or $i\leq k$. The proof is done separately in each of these cases. 

If $i>k$, then see that seller $i$ can not change the first $i-1$ elements of $\chi(F)$ by reporting a higher cost. Now, see that if $i>k+1$, then again $\chi_k$ will be chosen as the set of winners, which is a contradiction with the incentive of seller $i$ for misreporting. If $i=k+1$, then see that $\chi_{k}$ will be chosen as a subset of the winners; in this case, the set of winners can possibly contain other sellers, but not certainly not seller $k+1$ (due to the monotonicity of $F$). 
Consequently, seller $i$ remains a loser even by reporting $\over{c_i}$. Contradiction. 

It remains to do the proof for when $i\leq k$. Recall that the payment to seller $i$ is $2r_i \cdot \partial_i$. Since $r_i$ is not a function of the cost reported by seller $i$, then  see that the only way that $i$ can increase her utility by misreporting, is increasing $\partial_i$. But by reporting a cost higher than $c_i$, seller $i$ can not change the first $i-1$ elements of $\chi(F)$, which means she can only decrease $\partial_i$ by reporting a higher cost. So, the payment to $i$ does not increase if she reports a higher cost. This concludes the lemma.
\end{proof}

\begin{lemma} \label{lem.oracle2approx}
In a $\theta$-large market, the Oracle Mechanism has approximation ratio $\frac{1}{2}-\theta$, i.e. asymptotically equal to $\frac{1}{2}$ in a large market. 
\end{lemma}
\begin{proof}
It is enough to note that $F(\chi_{k+1})\geq F^\star /2$, which implies $F(\chi_{k})\geq F^\star\cdot (1/2-\theta)$ due to the large market assumption. 
\end{proof}

Now, we prove that sum of the payments in the Oracle Mechanism is at most $B+o(B)$, or in simple words, it is {\em almost budget feasible}.
\begin{definition}
A mechanism is {\em almost budget feasible} if its payments sum up to at most $B+o(B)$.
\end{definition}
As we mentioned before, (in large markets) we can convert any almost budget feasible mechanism to a budget feasible mechanism without any loss in its approximation ratio (asymptotically). This can be done simply by running the mechanism with a slightly reduced budget; the proof is deferred to Section \ref{sec.strictbf}.

\begin{lemma} \label{lem.oraclebudget}
The Oracle Mechanism is almost budget feasible.
\end{lemma}
\begin{proof}
Recall that the sum of payments is equal to $\sum_{i=1} ^k r_i \cdot \partial_i$.
 To prove the lemma, it is enough to show that
 for any seller $i\in S$, we have $r_i\leq \rstar \cdot (1-\theta)^{-1}$; because then we have
 \begin{align*}
\sum_{i=1} ^k r_i \cdot \partial_i \leq 
 \rstar\cdot (1-\theta)^{-1}\cdot\sum_{i=1} ^k \partial_i
 = \rstar \cdot (1-\theta)^{-1}\cdot F(\chi_k) \leq B(1-\theta)^{-1}.
\end{align*}
which proves the lemma.

To prove $r_i\leq \rstar \cdot (1-\theta)^{-1}$, let $F^\star_i$ denote the optimum utility when seller $i$ is removed. By the large market assumption, we have $\frac{F^\star_i}{F^\star}\geq 1-\theta$. This fact, and the fact that $r_i F^\star_i = \rstar F^\star = B$, imply that $r_i\leq \frac{\rstar}{1-\theta}$.
\end{proof}

\subsection*{The Oracle Mechanism in Polynomial Time}
In the Oracle Mechanism, we solve a submodular optimization problem which cannot be solved in polynomial-time, i.e. finding the optimum cost per utility rate which is equivalent to finding the optimum budget feasible subset. 
Although this problem cannot be solved in polynomial-time, there are approximation algorithms that can find near-optimal solutions for it.

\begin{definition}
Suppose we are given an instance of the problem with a submodular function $F$, budget $B$, and (publicly known) costs and utilities. A polynomial-time algorithm for solving this problem is called a {\em $\gamma$-approximation oracle} if, for any instance, it finds a solution with utility at least $\gamma \cdot F^\star$.
\end{definition}

Given a $\gamma$-approximation oracle, we can run the Oracle mechanism in polynomial time by finding estimates for $F^\star$ and $r_i$'s using the $\gamma$-approximation oracle, i.e. instead of computing $F^\star$ and $r_i$'s directly, we compute them using the $\gamma$-approximation oracle.
 The following theorem clarifies the resulting mechanism further and states the properties it satisfies.

\begin{theorem} \label{thm.gamma}
Suppose the Oracle Mechanism has access to a polynomial-time $\gamma$-approximation oracle for the knapsack optimization problem. Then, it is a polynomial-time $(\gamma/2)$-approximation mechanism and its payments sum up to at most $ (B+o(B))/\gamma$.
\end{theorem}
\begin{proof}
We use the $\gamma$-approximation oracle for finding an estimate (lower bound) for $F^\star$. Instead of computing $F^\star$ directly, let $F^\star$ be the solution which is returned by the $\gamma$-approximation oracle.
Also, for computing $r_i$, remove seller $i$ and then compute the optimal utility, namely $F^\star_i$, using the $\gamma$-approximation oracle. 
Everything else remains identical to Mechanism~\ref{alg.oracle}. 

The proofs for individual rationality and truthfulness follow from the proofs for Mechanism~\ref{alg.oracle}. It remains to prove that the payments sum up to at most $B\cdot (1+o(1))/\gamma$.
For this, we follow the proof of Lemma \ref{lem.oraclebudget}.
The key point in the proof of Lemma \ref{lem.oraclebudget} was that $r_i\leq \rstar \cdot (1-\theta)^{-1}$. Here, we prove a weaker inequality 
$r_i\leq\rstar \cdot (1-\theta)^{-1}/\gamma$. Given this inequality, the rest of the proof remains similar to the proof of Lemma \ref{lem.oraclebudget}. We do not repeat the full proof here and just show that the weaker inequality holds. 

Recall that in this proof, $F^\star$ and $F^\star_i$ denote the utilities computed by the $\gamma$-approximation oracle. Then, by the large market assumption, we have ${F^\star_i}\geq {{F^\star}}\cdot (1-\theta) \gamma$ . This fact, and the fact that $r_i F^\star_i = \rstar F^\star = B$, imply that $r_i\leq \rstar \cdot (1-\theta)^{-1}/\gamma$.
\end{proof}

\begin{corollary}[of Theorem \ref{thm.gamma}]
Suppose the Oracle Mechanism has access to a $\gamma$-approximation oracle. Then, 
if instead of budget $B$, the mechanism is given a reduced budget $\gamma B$, it would be an almost budget feasible mechanism with approximation ratio $\gamma^2/2$.
\end{corollary}

\subsection{An Improved Polynomial-Time Mechanism} \label{sec.polytime}
In this section, we present a polynomial-time mechanism with approximation ratio $\frac{1}{3}$. The mechanism follows the idea of the oracle mechanism, except that instead of computing the optimum cost per utility rate (which requires exponential running time), the mechanism computes an estimate for the cost per utility rate, which we call the {\em stopping rate}. Before presenting the mechanism, we need to formally define the notion of stopping rate.

\begin{definition}
Suppose we are given an instance of the problem with cost vector $c$, submodular utility function $F$ and budget $B$.
Construct the sequence $\chi(F)$ and choose the largest integer $k$ such that $F(\chi_k)\cdot \frac{c_{x_k}}{\partial_k} \leq B$. The {\em stopping rate}, denoted by $\r(c,F,B)$, is then defined to be $B/F(\chi_k)$. We sometimes denote the stopping rate simply by $\r(B)$ when $c,F$ are clear from the context. 
\end{definition}

Now we define the mechanism by presenting the winner selection and payment rules. 
\paragraph{Winner Selection}
The Mechanism constructs the sequence $\chi(F)$ and chooses the largest integer $k$ such that $F(\chi_k)\cdot \frac{c_{x_k}}{\partial_k} \leq B/2$. Then, it reports $\chi_k$ as the set of winners. 

\paragraph{The Payment Rule} 
For simplicity, assume that the winners are indexed from $1$ to $k$.
The Payment to winner $i$ is equal to $2 r_i \cdot \partial_i$, where $r_i = \r(c',F,B/2)$ and $c'$ is the cost vector which is identical to the original cost vector $c$ except that $c'_i$ is set to $\infty$ (i.e. intuitively, seller $i$ is removed from $S$). In other words, think of an auxiliary instance in which we are given budget $B/2$ and seller $i$ is removed from the set of sellers $S$. Then, $r_i$ is the stopping rate in this instance.

The above definitions are formally put together in Mechanism \ref{alg.ptoracle}.
We also address our polynomial-time mechanism as Mechanism \ref{alg.ptoracle}.
In the rest of this section, we prove that Mechanism \ref{alg.ptoracle} is individually rational, oneway-truthful and budget feasible, and also, it has approximation ratio $1/3$. As we mentioned before, this mechanism can be converted to a (fully) truthful and strictly budget feasible mechanism without any (asymptotic) loss in the approximation ratio; the details are discussed in Sections \ref{sec.o2t} and \ref{sec.strictbf}.

\begin{algorithm}
\LinesNotNumbered
\SetKwInOut{Input}{input}
\SetKwInOut{Output}{output}
\Input{Submodular utility function $F$, Budget $B$}
\BlankLine
Sellers report their costs\;
Construct the sequence $\chi(F)$\;
Find the smallest integer $k$ such that $F(\chi_k)\cdot\frac{c_{x_k}}{\partial_k}\leq B/2$\;
Announce $\chi_k$ as the set of winners\;
\ForEach{$i\in \chi_k$}{
	$c' \leftarrow c$\;
	$c'_i\leftarrow \infty$\;
	$r_i\leftarrow \r(c',F,B/2)$ \;
	Pay $r_i\cdot \partial_i$ to seller $i$\;
}
\caption{ A polynomial-time $(1/3)$-approximation mechanism{\label{alg.ptoracle}}}
\end{algorithm}

\paragraph{Simplifying Assumption} Through out the analysis, w.l.o.g. we assume that the sellers appear in the greedy sequence in an increasing order, i.e. $x_i=i$ for all $i\in S$. 

\begin{lemma}\label{lem.orind}
Mechanism \ref{alg.ptoracle} is individually rational. 
\end{lemma}
\begin{proof}
According to the payment rule, it is enough to show that for each winner $i$ we have $r_i\geq c_i/\partial_i$.
Let $\over{\chi}(F)=\langle \over{x}_1,\ldots,\over{x}_{n-1}\rangle$ denote the greedy sequence for when $i$ is removed. 
Also, let $\over{\chi}_j$ denote the subset containing first $j$ elements of $\over{\chi}(F)$, and $\over{\partial}_j=F(\over{\chi}_j)-F(\over{\chi}_{j-1})$. 

Note that $x_j=\over{x}_j$ for all $j<i$. 
Now, consider the following two cases for the proof:
Let $\over{k}$ be the largest integer $j$ satisfying $F(\over{\chi}_{j})\cdot {c}_{\over{x}_j}/\over{\partial}_{j}\leq B/2$. We either have that $\over{k}=i-1$ or $\over{k}\geq i$. 
We prove the lemma by proving  that $r_i\geq c_i/\partial_i$ in each of these cases.

First, suppose $\over{k}=i-1$. This means $r_i=\frac{B}{2F(\chi_{i-1})}$. Also, note that $F(\chi_{i})\cdot \frac{c_i}{\partial_i} \leq B/2$ (since $i < k$). The two latter facts, along with the fact that $F(\chi_{i-1})\leq F(\chi_{i})$ together imply that $r_i\geq c_i/\partial_i$.

It remains to do the proof in the second case: Suppose $\over{k}\geq i$. Note that due to the definition of the greedy sequence we have $c_{x_i}/\partial_i\leq c_{\over{x}_i}/\over{\partial}_i$. So, the proof is done if we show that $r_i \geq c_{\over{x}_i}/\over{\partial}_i$. 
To this end, first observe that we have $F(\over{\chi}_{i})\cdot {c_{\over{x}_i}}/{\over{\partial_i}} \leq B/2$ due to the fact that $\over{k}\geq i$.
This just implies ${c_{\over{x}_i}}/{\over{\partial_i}} \leq B/(2F(\over{\chi}_{\over{k}}))$ due to the monotonicity of $F$. Finally, seeing that $B/(2F(\over{\chi}_{\over{k}}))=r_i$ proves the claim.
\end{proof}

\begin{lemma}\label{lem.ptortruthful} 
Mechanism \ref{alg.ptoracle} is oneway-truthful. 
\end{lemma}
\begin{proof}
The proof is almost identical to the proof of Lemma \ref{lem.ortruthful}, we state the proof for completeness. 
For contradiction, suppose there exists a seller $i\in S$ who has incentive to report a cost $\over{c_i}$ which is higher than her true cost $c_i$. 
Assuming that the mechanism picked $\chi_k$ as the set of winners, we have either $i>k$ or $i\leq k$. The proof is done separately in each of these cases.

If $i>k$, then see that seller $i$ can not change the first $i-1$ elements of $\chi(F)$ by reporting a higher cost. Now, see that if $i>k+1$, then again $\chi_k$ will be chosen as the set of winners, which is a contradiction with the incentive of seller $i$ for misreporting. If $i=k+1$, then see that $\chi_{k}$ will be chosen as a subset of the winners; in this case, the set of winners can possibly contain other sellers, but not certainly not seller $k+1$ (due to the fact that she was not chosen before). 
Consequently, seller $i$ remains a loser even by reporting $\over{c_i}$. Contradiction. 

It remains to do the proof for when $i\leq k$. Recall that the payment to seller $i$ is $r_i \cdot \partial_i$. Since $r_i$ is not a function of the cost reported by seller $i$, then  see that the only way that $i$ can increase her utility by misreporting, is increasing $\partial_i$. But by reporting a cost higher than $c_i$, seller $i$ can not change the first $i-1$ elements of $\chi(F)$, and so, can only decrease $\partial_i$. Consequently, the payment to $i$ does not increase in this case as well which means $i$ has no incentive to report a higher cost.
\end{proof}

\begin{lemma}\label{lem.ptorratio} 
Mechanism \ref{alg.ptoracle} has approximation ratio $1/3$ in large markets. 
\end{lemma}
\begin{proof}
First we prove that $F(\chi_{k+1})\geq F^\star/3$. This would imply $F(\chi_{k})\geq F^\star\cdot (1/3-\theta)$ due to the large market assumption, which proves the lemma. 

So, it is enough to show that $F(\chi_{k+1})\geq F^\star/3$. We prove this by contradiction, assume $F(\chi_{k+1}) < F^\star/3$. 
 For more intuition, we first explain our argument for contradiction in words and then we state it more formally. 
From the definition of the greedy sequence $\chi(F)$, it can be seen that conditioned on buying the subset $\chi_{k+1}$, the cost for buying each extra unit of utility is at least $\frac{c_{k+1}}{\partial_{k+1}}$.
So, even if we get the subset $\chi_{k+1}$ for free, the cost for buying an additional $\frac{2 F^\star}{3}$ units of utility (which is needed for the optimum solution) would be at least
\begin{align*}
 \frac{2F^\star}{3}\cdot \frac{c_{k+1}}{\partial_{k+1}} >  
 2F(\chi_{k+1})\cdot \frac{c_{k+1}}{\partial_{k+1}} 
 \geq B. 
\end{align*}
This means cost of the optimum solution is more than $B$.

To formalize this contradiction, just note that by monotonicity of $F$, we have
\begin{align}
F({\sstar}\cup \chi_{k+1})-F({\chi_{k+1}}) > \frac{2 F^\star}{3}. \label{eq.lbscupchi11}
\end{align}
Now observe that by the definition of the greedy sequence $\chi(F)$, 
the cost for buying each extra unit of utility conditioned on having ${\chi_{k+1}}$ is at least $\frac{c_{k+1}}{\partial_{k+1}}$. This fact, and \eqref{eq.lbscupchi11} together imply that 
\begin{align}
c (\sstar\cup \chi_{k+1}) - c (\chi_{k+1})> 
\frac{c_{k+1}}{\partial_{k+1}}\cdot \frac{2 F^\star}{3}. \label{eq.12}
\end{align}
On the other hand, recall that $F(\chi_{k+1}) < F^\star/3$, so we can write \eqref{eq.12} as
\begin{align*}
c (\sstar\cup \chi_{k+1}) - c (\chi_{k+1}) >\,  & 
\frac{c_{k+1}}{\partial_{k+1}}\cdot \frac{2 F^\star}{3}\\
> \, &
 2F(\chi_{k+1})\cdot \frac{c_{k+1}}{\partial_{k+1}} 
> B. 
\end{align*}
This implies $c (\sstar) > B$ which is a contradiction with the budget feasibility of $\sstar$.
\end{proof}

We prove that Mechanism \ref{alg.ptoracle} is almost budget feasible in Lemma \ref{lem.ptorbfeasible} . Before that, we first need to prove the following lemma which will be used in the proof of Lemma \ref{lem.ptorbfeasible}. This Lemma states that $F(\chi_k)$ is (roughly) at least half of the optimal utility in a large market.

\begin{lemma}\label{lem.xklarge} 
Suppose we are given a $\theta$-large market with $\theta\leq 1/2$ and let $\chi_k$ denote the subset chosen by Mechanism \ref{alg.ptoracle} when run on this instance. Then we have 
\begin{align*}
F(\chi_k)\geq ({1}/{2}- \theta)\cdot F^\star
\end{align*}
\end{lemma} 
\begin{proof}
Recall that w.l.o.g. we assume $\chi=\langle1,\ldots,n\rangle$.
To prove the lemma, it is enough to show that $F(\chi_{k+1})\geq F^\star/2$.
\begin{align*}
F(\chi_{k})\geq F(\chi_{k+1})-\theta\cdot F^\star \geq ({1}/{2}- \theta)\cdot F^\star,
\end{align*}
where the first inequality is due to the large market assumption.

We prove $F(\chi_{k+1})\geq F^\star/2$ by contradiction, assume $F(\chi_{k+1}) < F^\star/2$. For more intuition, we first explain our argument for contradiction in words and then we state it more formally.
Conditioned on buying the subset $\chi_{k+1}$, the cost for buying each extra unit of utility is at least $\frac{c_{k+1}}{\partial_{k+1}}$.
So, even if we get the subset $\chi_{k+1}$ for free, the cost for buying an additional $\frac{F^\star}{2}$ units of utility (which is needed for the optimum solution) would be strictly more than $\frac{F^\star}{2}\cdot \frac{c_{k+1}}{\partial_{k+1}}\geq B$. To formalize this contradiction, just note that by monotonicity of $F$, we have
\begin{align}
F({\sstar}\cup \chi_{k+1})-F({\chi_{k+1}}) > \frac{F^\star}{2}. \label{eq.lbscupchi}
\end{align}

Now observe that by the definition of the greedy sequence $\chi(F)$, 
the cost for buying each extra unit of utility conditioned on having ${\chi_{k+1}}$ is at least $\frac{c_{k+1}}{\partial_{k+1}}$. This fact, and \eqref{eq.lbscupchi} together imply that 
\begin{align*}
c (\sstar\cup \chi_{k+1}) - c (\chi_{k+1})> \frac{c_{k+1}}{\partial_{k+1}}\cdot \frac{F^\star}{2} \geq B. 
\end{align*}
This implies $c (\sstar) > B$ which is a contradiction with the budget feasibility of $\sstar$.

\end{proof}

\begin{lemma}\label{lem.ptorbfeasible} 
Mechanism \ref{alg.ptoracle} is almost budget feasible
\end{lemma}
\begin{proof}
Let $r$ denote $\r(c, F, B/2)$, which is equal to $\frac{B/2}{F(\chi_k)}$ by definition.
We show that for any $i\in \chi_{k}$, we have $r_i \leq 2 r\cdot (1-3\theta)^{-1}$. This would prove that sum of the payments is at most 
$B\cdot (1-3\theta)^{-1}$ since:
\begin{align*}
\sum_{i=1}^k r_i \partial_i \leq (1-3\theta)^{-1}\cdot\sum_{i=1}^k 2 r\partial_i =  (1-3\theta)^{-1}\cdot 2 r F(\chi_k) = (1-3\theta)^{-1}\cdot B.
\end{align*}
This would finish the proof due to the large market assumption. 

To this end, fix a seller $i\in S$ and let $\over{\chi}(F)=\langle \over{x}_1,\ldots,\over{x}_{n-1}\rangle$ denote the greedy sequence for when $i$ is removed from the set of sellers $S$. 
Also, let $\over{\chi}_j$ denote the subset containing first $j$ elements of $\over{\chi}(F)$, and $\over{\partial}_j=F(\over{\chi}_j)-F(\over{\chi}_{j-1})$. 
Finally, let $\over{k}$ be the largest integer $j$ satisfying $F(\over{\chi}_{j})\cdot {c}_{\over{x}_j}/\over{\partial}_{j}\leq B/2$. 

Suppose $F^\star_{B/2}$ denotes the optimal utility for when the budget is reduced to $B/2$. Then we would have:
\begin{align}
r_i =  \frac{B/2}{F(\over{\chi}_{\over{k}})} \leq &\frac{B/2}{ (1/2-\theta)\cdot F^\star_i } \label{eq.use.lem.xklarge}\\ 
\leq & \frac{B/2}{ (1/2-\theta)\cdot (1-\theta)\cdot F^\star_{B/2} } \label{eq.use.lrgmrkt}\\ 
\leq & \frac{B/2}{ (1/2-\theta)\cdot (1-\theta)\cdot F(\chi_k) } \nonumber\\
\leq & \frac{r}{ (1/2-\theta)\cdot (1-\theta)} \nonumber \\
\leq & \frac{2r}{ 1-3\theta } \nonumber
\end{align}
where \eqref{eq.use.lem.xklarge} is due to Lemma \ref{lem.xklarge} and \eqref{eq.use.lrgmrkt} is due to the $\theta$-largeness of the market.
\end{proof}

\subsection{From Oneway-Truthfulness to Truthfulness}\label{sec.o2t}
Changing the payment rule is the key to convert all of our oneway-truthful mechanisms to truthful mechanisms. In fact, the winner selection rule remains identically the same, however, we replace the payment rule in all of our mechanisms by the following simple payment rule:

\paragraph{The Critical Cost Payment Rule} Each winner $i$ is paid the highest cost that it could report and still remain a winner.
\\

Individual rationality and truthfulness of all the mechanisms 
are trivial under this payment rule. Also, the proofs for approximation ratio remain the same since we have not changed the winner selection rule. The only non-trivial part is proving the budget feasibility under this new payment rule. For all of the mechanisms in this paper, we can show that replacing the original payment rule with the Critical Cost payment rule, does not increase the payment of the mechanism to any of the winners. 
First we prove this for the Oracle mechanism.

\begin{lemma} \label{lem.oracleccbf}
In the Oracle mechanism, the payment to each winner does not increase if we replace the payment rule with the Critical Cost payment rule.
\end{lemma}
\begin{proof}  
We need to show that the Critical Cost payment rule does not pay to winner $i$ more than $2r_i\cdot \partial_i$. Equivalently, we show that if $i$ reports a cost $\over{c_i}>2r_i\cdot \partial_i$, then it will not be selected as a winner anymore. 
The proof is by contradiction, suppose $i$ reports a cost higher than $\over{c_i} > 2r_i\cdot \partial_i$ and still gets selected as a winner. 

Let the instance in which the cost $c_i$ is replaced by $\over{c_i}$ be called the {\em fake} instance. 
Let $\over{c}$ denote the cost vector in the fake instance which is identical to the original cost vector $c$ except that $c'_i$ is set to $\over{c_i}$. For any subset of sellers $S'\subseteq S$, let $\over{c}(S') = \sum_{s\in S'} \over{c}_s$.

Also, let $\over{\chi}(F)$ denote the greedy sequence for the fake instance, $\over{\chi}_j$ denote the subset containing first $j$ elements of $\over{\chi}(F)$, and $\over{\partial}_j=F(\over{\chi}_j)-F(\over{\chi}_{j-1})$. Finally, let $\over{\sstar}$ denote the optimum budget feasible subset in the fake instance.

For more intuition, we first explain our argument for contradiction in words and then we state it more formally. 
See that the optimum utility is at least $F^\star_i$ in the fake instance (recall that $F^\star_i$ denotes the optimum utility when $i$ is removed from $S$), which implies 
$F(\over{\sstar})-F(\over{\chi}_{\over{k}})\geq \frac{F^\star_i}{2}$.
Also, we will show that even if we get $\over{\chi}_{\over{k}}$ for free, the cost for buying an additional $F(\over{\sstar})-F(\over{\chi}_{\over{k}})$ units of utility would be more than $2 r_i$ {\em per unit of utility}. The two latter facts together would imply that 
\begin{align*}
\over{c} (\over{\sstar}\backslash \over{\chi}_{\over{k}}) > 
2 r_i \cdot \left( F(\over{\sstar})-F(\over{\chi}_{\over{k}}) \right)
\geq  2 r_i\cdot \frac{F^\star_i}{2} = B
\end{align*}
which implies $c (\over{\sstar}) > B$. This is a contradiction with budget feasibility of $\over{\sstar}$.

To formalize this contradiction, just note that by the monotonicity of $F$, we have
\begin{align}
F(\over{\sstar}\cup \over{\chi}_{\over{k}})-F(\over{\chi}_{\over{k}})\geq \frac{F^\star_i}{2}.\label{eq.lbscupchi112}
\end{align}
Now observe that by the definition of the greedy sequence $\over{\chi}(F)$, the cost for buying each extra unit of utility, conditioned on having the subset $\over{\chi}_{\over{k}}$, 
is at least ${\over{c_i}}/{\over{\partial}_{\over{k}}}$, i.e. the cost of the seller at position $\over{k}$ of $\over{\chi}(F)$ divided by the marginal utility that it adds. Also, see that ${\over{c_i}}/{\over{\partial}_{\over{k}}} > 2 r_i$ since ${\over{c_i}}> 2r_i\cdot\partial_i$ and ${\over{\partial}_{\over{k}}}\leq \partial_i$. This proves that the cost for buying each extra unit of utility, conditioned on having the subset $\over{\chi}_{\over{k}}$, is more than $2 r_i$. 
This fact, and \eqref{eq.lbscupchi112} together imply that 
\begin{align*}
\over{c} (\over{\sstar}\cup \over{\chi}_{\over{k}}) - \over{c} (\over{\chi}_{\over{k}}) > 
2 r_i\cdot \frac{F^\star_i}{2} = B.
\end{align*}
This implies $\over{c} (\over{\sstar}) > B$ which is a contradiction with the budget feasibility of $\sstar$.
\end{proof}

Now, we prove an equivalent version of Lemma \ref{lem.oracleccbf} for Mechanism \ref{alg.ptoracle}. 
\begin{lemma}
In Mechanism \ref{alg.ptoracle}, the payment to each winner does not increase if we replace the payment rule with the Critical Cost payment rule.
\end{lemma}
\begin{proof}  
We need to show that the Critical Cost payment rule does not pay to winner $i$ more than $r_i\cdot \partial_i$. Equivalently, we show that if $i$ reports a cost $x$ with $x > r_i\cdot \partial_i$, then it will not be selected as a winner anymore. 
The proof is by contradiction, suppose $i$ reports a cost $x$ which is higher than $r_i\cdot \partial_i$ and still gets selected as a winner. 

Let the instance in which the cost $c_i$ is replaced by $x$ be called the {\em fake} instance. 
Let $\over{c}$ denote the cost vector in the fake instance which is identical to the original cost vector $c$ except that $\over{c}_i$ is set to $x$. For any subset of sellers $S'\subseteq S$, let $\over{c}(S') = \sum_{s\in S'} \over{c}_s$; when there is no risk of confusion, we sometimes denote $c(\{i\})$ by $c(i)$ for any cost function $c(\cdot)$. 

Also, consider the instance in which the cost $c_i$ is replaced by $\infty$ (seller $i$ is removed) and call it the {\em large} instance. Let $\over{\over{c}}$ denote the cost vector in the large instance which is identical to the original cost vector $c$ except that $\over{\over{c}}_i$ is set to $\infty$. For any subset of sellers $S'\subseteq S$, let $\over{\over{c}}(S') = \sum_{s\in S'} \over{\over{c}}_s$.

Let $\chibar(F)=\langle \over{x}_1,\ldots, \over{x}_{n}\rangle$ and $\chidbar(F)=\langle \over{\over{x}}_1,\ldots, \over{\over{x}}_{n-1}\rangle$ respectively denote the greedy sequences for the fake and large instances, and $\over{\chi}_j , \chidbar_j$ respectively denote the subsets containing the first $j$ elements of $\over{\chi}(F)$ and $\chidbar(F)$. Also, let $\over{\partial}_j=F(\over{\chi}_j)-F(\over{\chi}_{j-1})$ and $\over{\over{\partial}}_j=F(\chidbar_j)-F(\chidbar_{j-1})$. 

Suppose that Mechanism \ref{alg.ptoracle} picks $\chibar_{\over{k}}$ as the set of winners in the fake instance and $\chidbar_{\over{\over{k}}}$ as the set of winners in the large instance, i.e. the mechanism picks the first $\over{k}$ elements of $\chibar(F)$ and the first $\over{\over{k}}$ elements of $\chidbar(F)$ as the set of winners in each of the instances.  
First, we need to prove the following claim and then we proceed to the proof of the lemma. 

\begin{claim} \label{clm.subset}
The first $\over{\over{k}}$ elements of $\chibar(F)$ and $\chidbar(F)$ are identical. 
\end{claim}
\begin{subproof}
Let $\over{{i}}$ denote the position of seller $i$ in $\chibar(F)$, i.e. we have $\over{x}_{\over{i}}=i$. Then, see that 
\begin{align}
c({{\over{\over{x}}}_{{\over{\over{k}}}}})/\over{\over{\partial}}_{\over{\over{k}}}  \leq r_i, \label{eq.subset1}
\end{align}
which holds since we have $F(\chidbar_k)\cdot c({{\over{\over{x}}}_{{\over{\over{k}}}}})/\over{\over{\partial}}_{\over{\over{k}}} \leq B/2$ by the definition of $\over{\over{k}}$ and $r_i = B/(2 F(\chidbar_{\over{\over{k}}}))$ by the definition of stopping rate.

On the other hand, see that 
\begin{align}
\over{{c}} ({x_{{\over{{i}}}}}) /\over{{\partial}}_{{\over{i}}} > & r_i \cdot \partial_i / \over{{\partial}}_{{\over{i}}} \nonumber \\
\geq & r_i \label{eq.321}
\end{align}
where \eqref{eq.321} is due to the fact that 
$\partial_i \geq  \over{{\partial}}_{{\over{i}}}$; this fact holds because the first $i-1$ elements of $\chi(F)$ and $\chibar(F)$ are identical, which means the marginal utility added by seller $i$ in $\chi(F)$ is more than her marginal in $\chibar(F)$.

To summarize, see that by \eqref{eq.subset1} and \eqref{eq.321} we have:
\begin{align*}
c({{\over{\over{x}}}_{{\over{\over{k}}}}})/\over{\over{\partial}}_{\over{\over{k}}}  \leq & \,r_i,\\
\over{{c}} ({\over{x}_{{\over{{i}}}}}) /\over{{\partial}}_{{\over{i}}} > &\,  r_i 
\end{align*}
which means seller $i$ will be never be used in its first $\over{\over{k}}$ positions of $\chibar(F)$ (due to the greedy construction of the sequence). 
Consequently, the first $\over{\over{k}}$ elements of $\chibar(F)$ and $\chidbar(F)$ are identical.
 \end{subproof}

Now we are ready to see to the contradiction, it follows from the following set of inequalities which will be clarified below.
\begin{align}
B/2 \geq\, & F(\chibar_{\over{k}}) \cdot \over{c}_{\over{k}}/\over{\partial_{\over{k}}} \label{eq.p1}\\
> \, & F(\chibar_{\over{k}}) \cdot \frac{r_i\cdot \partial_i}{\over{\partial_{\over{{i}}}}} \label{eq.p2}\\
\geq\, & F(\chibar_{\over{k}}) \cdot r_i  \label{eq.p3}\\
\geq\, & F(\chidbar_{\over{\over{k}}}) \cdot r_i \label{eq.p4}\\
= & B/2 \nonumber
\end{align}
\eqref{eq.p1} is due to the definition of $\over{k}$ in the fake instance; \\
\eqref{eq.p2} holds since we have assumed $\over{i}\leq \over{k}$ (i.e. $i$ wins in the fake instance) and so we have $\over{c}({\over{x}}_{\over{i}})/\over{\partial_{\over{i}}} \leq \over{c}_{\over{k}}/\over{\partial_{\over{k}}} $ by the definition of the greedy sequence, this fact, and the fact that $\over{c}({\over{x}}_{\over{i}}) > r_i \cdot \partial_i$ directly imply \eqref{eq.p2}; \\
 \eqref{eq.p3} holds since $\partial_i \geq \over{\partial}_{\over{i}}$; we already proved this is true in the proof of Claim \ref{clm.subset}: because the first $i-1$ elements of $\chi(F)$ and $\chibar(F)$ are identical;\\
\eqref{eq.p4} follows by the monotonicity of $F$ and since $\chidbar_{\over{\over{k}}} \subseteq \chibar_{\over{k}}$ (which holds by  Claim \ref{clm.subset}).
\end{proof}

\subsection{Strictly Budget Feasible Mechanisms} \label{sec.strictbf}

In this section, we show how to convert our almost budget feasible mechanisms to strictly budget feasible mechanisms. We start with the Oracle Mechanism, i.e. Mechanism \ref{alg.oracle}. Recall that in Lemma \ref{lem.oraclebudget}, we proved that sum of the payments in the Oracle mechanism is at most $B\cdot (1-\theta)^{-1}$. So, if instead of budget $B$, we give a reduced budget $B\cdot (1-\theta)$ to the oracle mechanism, then the sum of its payments would not exceed $B$. 

It only remains to show that this budget reduction does not affect the approximation ratio significantly. To this end, first we prove that the budget reduction does not affect the optimum utility significantly.

\begin{lemma}\label{lem.optboverk} 
Assume we are given a $\theta$-large market 
and Let $F^\star_{b}$ denote the optimal solution for this market when 
when the budget is reduced to $b\leq B$. Then for any given constant $\epsilon>0$ and $b=B\cdot (1-\epsilon)$ we have:
\begin{align*}
F^\star_{b}\geq (1-\theta-\epsilon)\cdot F^\star\end{align*}
\end{lemma}
\begin{proof}
W.L.O.G. suppose that with budget $B$, the optimal subset is $\sstar=\{1,\ldots,s\}$.
Let the submodular function $G$ denote the restriction of $F$ to the subset $\sstar$, i.e. $G=F|_{\sstar}$. 
Now, construct the greedy sequence $\chi(G)$, and denote it by $\chi(G)=\langle 1,\ldots, s\rangle$, were w.l.o.g. we have assumed that the members of $\sstar$ appear in the greedy sequence in the increasing order.

By assumption, we have that $c(\sstar) \leq B$. Let $s'$ denote the smallest integer such that 
\begin{align*}
\sum_{i=1}^{s'} c_i \geq B\cdot (1-\epsilon). 
\end{align*}
Also, let $S'=\{1,\ldots,s'-1\}$. We claim that 
\begin{align*}
F(S')\geq  (1-\theta-\epsilon)\cdot F^\star,
\end{align*}
 which proves the lemma since $c(S')\leq B\cdot (1-\epsilon)$. To this end, first verify that 
\begin{align*} 
 F(S'\cup\{s'\})\geq F^\star\cdot (1-\epsilon);
\end{align*} 
this holds by the definition of the greedy sequence, which picks the cheaper items (relative to their utility) first. Then, observe that $F(S')\geq F(S'\cup\{s'\})-F(\{s'\})$ due to submodularity. This implies 
\begin{align*}  
 F(S')\geq  (1-\epsilon)\cdot F^\star -\theta \cdot F^\star = F^\star \cdot (1-\theta-\epsilon)
\end{align*}  
 due to the large market assumption.  
\end{proof}

To achieve a strictly budget feasible mechanism for large markets, fix $\epsilon$ to be an arbitrary small constant. In a large market with $\theta < \epsilon$, we can reduce the budget of the Oracle mechanism to $B\cdot(1-\epsilon)$ to get a strictly budget feasible mechanism: the Oracle Mechanism would be budget feasible by Lemma \ref{lem.oraclebudget}. Also, by Lemma \ref{lem.optboverk}, such a market (with the reduced budget) would be $\frac{\theta}{1-\theta-\epsilon}$-large. 
So, by Lemmas \ref{lem.oracle2approx} and \ref{lem.optboverk}, the approximation ratio of the Oracle mechanism would be
\begin{align*} 
\left( \frac{1}{2}- \frac{\theta}{1-\theta-\epsilon}\right)\cdot (1-\theta-\epsilon).
\end{align*}  
See that in large markets, where $\theta\to 0$, the approximation ratio would be 
\begin{align*} 
\lim_{\theta\to 0}\left( \frac{1}{2}- \frac{\theta}{1-\theta-\epsilon}\right)\cdot (1-\theta-\epsilon). = \frac{1}{2} \cdot ({1-\epsilon}),
\end{align*}  
 i.e. for any arbitrary small $\epsilon>0$, we have a large market mechanism with approximation ratio $\frac{1}{2}\cdot (1-\epsilon)$.

By a similar argument, it can be seen that for any arbitrary small constant $\epsilon > 0$, Mechanism~\ref{alg.ptoracle} can also be converted to a strictly budget feasible mechanism with approximation ratio $\frac{1}{3}\cdot(1-\epsilon)$ in large markets.

\section{Hoeffding Bounds}\label{sec.hoeffding}
In this section we state a version of Hoeffding bounds \citepalias{wikihoeffding} that is suitable for our purpose.  

\paragraph{Hoeffding Bounds.} Let $x_1,\ldots,x_n$  be i.i.d. random variables such that $\Pr(x_i\in [a,b])=1$. Let $\mu = \bbE\left[\sum_{i=1}^n x_i\right]$, then we have:
\begin{align*}
\Pr\left(\sum_{i=1}^n x_i \geq (1+\epsilon)\cdot \mu\right) \leq e^{-\frac{2 \epsilon^2 \mu^2}{n(b-a)^2}}.
\end{align*}

\end{document}